\documentclass[journal,11pt, draftclsnofoot, onecolumn]{IEEEtran}

\ifCLASSINFOpdf
\else
\fi

\usepackage{times,amsmath}
\usepackage{amsfonts}
\usepackage{amssymb,bm}
\usepackage{bbm}
\usepackage{algorithm}
\usepackage{algpseudocode}
\usepackage{cite}
\usepackage{graphicx}
\usepackage{epstopdf}
\usepackage[caption=false,font=footnotesize]{subfig}
\usepackage{fixltx2e}
\usepackage{balance}
\usepackage{cases}
\usepackage{color}
\usepackage{enumerate}

\allowdisplaybreaks

\usepackage{amsthm}
\newtheorem{theorem}{Theorem}

\newtheorem{lemma}{Lemma}
\newtheorem{remark}{Remark}

\newtheorem{corollary}{Corollary}
\usepackage{chngcntr}
\counterwithout{theorem}{section}
\counterwithout{definition}{section}
\counterwithout{lemma}{section}
\counterwithout{remark}{section}
\counterwithout{assumption}{section}
\counterwithout{proposition}{section}
\counterwithout{corollary}{section}

\DeclareMathOperator{\E}{\mathbbmss{E}}
\DeclareMathOperator{\Prob}{\mathbbmss{P}}
\DeclareMathOperator{\Trn}{\mathsf{T}}

\IEEEoverridecommandlockouts
\begin{document}
\bstctlcite{IEEEexample:BSTcontrol}
%
\title{Generalized Degrees of Freedom of the Symmetric Cache-Aided MISO Broadcast Channel with Partial CSIT}
\author{Enrico~Piovano, Hamdi~Joudeh and Bruno~Clerckx
\thanks{The authors are with the Communications and Signal Processing group, Department of Electrical and Electronic Engineering, Imperial College London, London SW7 2AZ, U.K. (email: \{e.piovano15; hamdi.joudeh10; b.clerckx\}@imperial.ac.uk).}
\thanks{This work was partially supported by the U.K. Engineering and Physical Sciences Research Council (EPSRC) under grant EP/N015312/1.
This paper was presented in part at the 2018 IEEE International Symposium on
Information Theory \cite{Piovano2018}.}}
\maketitle
\begin{abstract}
We consider the cache-aided MISO broadcast channel (BC) in which a multi-antenna transmitter
serves $K$ single-antenna receivers, each equipped with a cache memory.
The transmitter has access to partial knowledge of the channel state information.
For a symmetric setting, in terms of channel strength levels, partial channel knowledge levels and cache sizes,
we characterize the generalized degrees of freedom (GDoF) up to a constant multiplicative factor.
The achievability scheme exploits the interplay between spatial multiplexing gains and coded-multicasting gain.
On the other hand, a cut-set-based argument in conjunction with a GDoF outer bound for a parallel MISO BC under channel uncertainty are used for the converse.
We further show that the characterized order-optimal GDoF is also attained in a decentralized setting, where no coordination is required for content placement in the caches.
\end{abstract}
\section{Introduction}
Traffic over wireless networks is predominantly becoming content-oriented, a transformation mainly driven by the advent of multimedia applications, especially video-on-demand services \cite{Bastug2014}.
For this type of traffic, there is often a large content library out of which  users request specific files.
The content library is typically generated well before transmission, creating the opportunity to pre-store (i.e. cache) parts of
the content at different nodes across the network during off-peak times, when the network resources are under utilized.
This cached information is then used during peak times, when users are actively requesting content and competing for wireless spectrum, to reduce the transmission load over the  network \cite{Shanmugam2013}.
Therefore, such cache-aided networks often operate in two phases: a \emph{placement phase} which takes place during off-peak times, and a \emph{delivery phase} which takes place during peak times \cite{Maddah-Ali2014}.

In single-user systems, the caching gain comes from making
part of the content locally available to the user.
Such \emph{local caching gain} scales with the cache memory size, and extends to networked systems with no interference, i.e.
where each user enjoys a dedicated and isolated communication link.
The picture, however, is very different when users share communication links.
This was taken up by Maddah-Ali and Niesen in \cite{Maddah-Ali2014}, where caching was investigated in the context of a broadcast network in which one transmitter (server) communicates with multiple users, equipped with cache memories, over a shared noiseless link.
In addition to the obvious local caching gains, Maddah-Ali and Niesen revealed a (hidden) \emph{global caching gain} which scales with the aggregate size of cache memories distributed across the network, despite the lack of cooperation amongst users during transmissions.
Such global caching gain is exploited through careful placement of content during the placement phase, creating
(coded) multicasting opportunities during the delivery phase, that would not naturally occur otherwise.
This in turn allows serving multiple distinct user demands using fewer transmissions.

Global caching gains were initially demonstrated assuming a centralized setting, were centrally coordinated placement takes place \cite{Maddah-Ali2014}.
While the placement phase takes place during off-peak hours before user demands are known to the transmitter, it was still assumed that it was carried out in a centrally coordinated manner in which
 the number and identity of active users during the delivery phase are known beforehand.
This is often difficult to satisfy in practical networks, particularly in wireless settings where users enjoy a high degree of mobility.
This called for developing a decentralized version of coded-caching, where placement is randomized and hence
independent of the identity and number of active users during the delivery phase \cite{Maddah-Ali2015a}. Surprisingly, it was shown in \cite{Maddah-Ali2015a} that decentralization comes at a low price, achieving an order-optimal performance comparable to the centralized scheme.

The coded-caching framework above has been further extended in many directions.
Such developments were recently surveyed in \cite{Maddah-Ali2016}, in which challenges and open problems are also discussed.
One of the main open problems identified in \cite{Maddah-Ali2016}
is the capacity characterization of cache-aided wireless networks.
\subsection{Cache-Aided Wireless Networks}
The capacity of wireless networks is one of the longest standing open problems in network information theory.
The intractability of the problem, in its generality, motivated the use of capacity approximations, e.g. the Degrees of Freedom (DoF) metric and the Generalized Degrees of Freedom (GDoF) metric.
The introduction of such metrics allowed significant progress in capacity studies.
Since incorporating caches adds an extra layer of complexity to the network, it is not surprising to see
that the utilization of the above approximations is inherited by works studying cache-aided wireless networks.
Examples of such studies in different scenarios are given in
\cite{Maddah-Ali2015,Naderializadeh2017,Hachem2018,Xu2017,Sengupta2017,Kakar2017,Ji2016,Yi2016,Zhang2015,Zhang2017,Lampiris2017,Piovano2017}.

Amongst the main insights derived from the above studies is that caching at the transmitters creates interference alignment and zero-forcing opportunities, enabled through partial and full transmitter cooperation.
For example, interference channels start resembling X channels and eventually turn into multi-antenna broadcast channels \cite{Maddah-Ali2015,Naderializadeh2017,Hachem2018,Xu2017}.
On the other hand, caching at receivers creates coded-multicasting opportunities,
which are particularly useful in scenarios where spatial degrees of freedom cannot sufficiently create parallel interference free links.
For example, coded-multicasting gains are pronounced in multi-antenna broadcast channels
with more receivers than transmitting antennas \cite{Naderializadeh2017,Shariatpanahi2017}
and/or where channel state information at the transmitter (CSIT) is imperfect \cite{Zhang2015,Zhang2017,Lampiris2017,Piovano2017}.
\subsection{The Cache-Aided MISO Broadcast Channel}
{In this paper, we focus on the cache-aided multiple-input-single-output broadcast channel (MISO BC),
in which a $K$-antenna transmitter serves $K$ single-antenna users, where each user is equipped with a cache memory.
Note that the $K$ transmit antennas in the considered setup are not necessarily physically co-located,
and may generally represent $K$ radio heads (or remote antennas) connected through a strong
fronthaul.}
When CSIT is available with high accuracy,
parallel non-interfering links can be created through zero-forcing.
In this case, interference is completely managed through spatial pre-processing, and the usefulness of caches is restricted to local caching gains.
However, this is not the case when only partial of imperfect CSIT is available as observed  in \cite{Zhang2015,Zhang2017,Lampiris2017}.

Studying the classical MISO BC (with no caches) reveals that spatial multiplexing gains (i.e. DoF) of this channel suffer losses under imperfect CSIT.
For example, the extreme case of finite precision CSIT causes a total collapse of the DoF to $1$, where all (DoF) benefits of multiple transmitting antennas are lost \cite{Davoodi2016}.
The availability of partial instantaneous CSIT can help salvage some of the lost gains, achieving DoF between $1$ and $K$ depending on the CSIT quality.
The complementary role of coded-caching in such scenarios was first observed in \cite{Zhang2015}.
In particular, while the primary role of CSIT is to facilitate interference management (e.g. through zero-forcing),
coded-caching reduces interference all together by creating multicasting opportunities.
Hence, it was shown in \cite{Zhang2015} that coded-caching can offset the loss due to partial CSIT, up to a certain CSIT quality given the cache size.

The DoF metric, however, can be very pessimistic, as best exemplified by the DoF collapse in \cite{Davoodi2016}.
This is mainly due the limitations of the DoF framework, assigning equal strengths to every link (with non-zero gain) in the wireless network.
In a way, the DoF metric fails to capture one of the wireless channel's most important features: propagation loss.
This limitation is countered by the GDoF framework, which largely inherits the tractability of the DoF framework
while capturing the diversity in channel strengths \cite{Etkin2008,Davoodi2017a,Davoodi2018}.
The cache-aided MISO BC was studied under the GDoF framework in \cite{Lampiris2017}, while limiting to completely absent CSIT
and considering only achievability, with no guarantees on optimality\footnote{The same can be said about \cite{Zhang2015}, where the DoF under partial CSIT can be equivalently interpreted as the GDoF under no CSIT (see Section \ref{subsubsec:GDoF_no_CSIT}). No converse is given in \cite{Zhang2015}, except for the trivial case where perfect CSIT is available.}.
In a different line of work, the cache-aided MISO BC under partial CSIT was considered while focusing on the massive MIMO regime \cite{Ngo2018}.
In particular, \cite{Ngo2018} studies the delivery rate scaling laws, as the number of transmitting antennas grows arbitrarily large, using off-the-shelf caching strategies.
While no guarantees on information-theoretic optimality are provided in the above work,
the emphasis on the interplay between spatial multiplexing gains and coded-multicasting gains is very interesting.
It turns out that this interplay, which was first noticed in \cite{Zhang2015} and then further investigated in \cite{Piovano2017,Shariatpanahi2017,Ngo2018}, plays a central role in achieving and interpreting the order-optimal GDoF of the cache-aided MISO BC under partial CSIT as we show through our results.
Next, we highlight the main contribution of this paper.
\subsection{Main Contributions and Organization}
We consider a $K$-user cache-aided MISO BC within the (symmetric) GDoF framework,
where the channel strength of cross-links is captured through the famous $\alpha \in [0,1]$ parameter \cite{Etkin2008,Davoodi2017a,Davoodi2018}.
In addition, we capture the entire range of (symmetric) partial CSIT levels through the quality parameter $\beta \in [0,\alpha]$,
where $\beta = 0$ and $\beta = \alpha$ correspond to essentially absent and perfect CSIT, respectively \cite{Davoodi2018}.
For this setting, the main contributions are twofold, as stated below:
\begin{enumerate}
\item We characterize the optimal GDoF up to a constant multiplicative factor, which is independent of all system parameters.
This order-optimal GDoF characterization is derived while allowing central coordination during the placement phase of the
achievability scheme.
\item We show that the order-optimal GDoF, characterized under centralized placement, is also attained in
decentralized settings where no coordination during the placement phase is allowed.
\end{enumerate}
It is worthwhile highlighting that the order optimal schemes
for the considered cache-aided MISO BC, for both the centralized and decentralized cases,
abide by the \emph{separation} principle \cite{Naderializadeh2017a}.
In particular, the placement and generation of coded-multicasting messages are independent
of the physical channel parameters (e.g. link strengths or topology), and follow the placement and message generation
of the original shared-link Maddah-Ali and Niesen schemes \cite{Maddah-Ali2014,Maddah-Ali2015a}.
On the other hand, the delivery of the coded-multicasting messages over the physical channel uses the principle of
rate-splitting with common and private signalling, commonly employed for the classical MISO BC with partial CSIT \cite{Yang2013,Joudeh2016,Davoodi2018},
and essentially operates the physical channel at some point of its multiple multicast GDoF region.

One of the technical challenges in characterizing the optimal GDoF for the above setting is the converse, i.e.
deriving an outer bound which is within a constant multiplicative factor from the achievable GDoF.
Under partial CSIT, the conventional cut-set-based argument in \cite{Maddah-Ali2014}
fails when employed on its own (see also \cite{Hachem2018,Xu2017,Sengupta2017} for variants of such argument).
Alternatively, we derive an outer bound by marrying the approach in \cite{Maddah-Ali2014} with
a robust GDoF outer bound for a parallel MISO BC under partial CSIT,
which in turn employs results from recent works by Davoodi and Jafar on classical
networks (with no caches) under  channel uncertainty
\cite{Davoodi2016,Davoodi2017a,Davoodi2018}.
Specifically, in this novel adaptation of the approach in \cite{Davoodi2016,Davoodi2017a,Davoodi2018} to cache-aided network,
caches at receivers are replaced with equivalent parallel side links, and then an upper bound on the GDoF of the resulting parallel sub-channels is derived.

Another technical challenge arises when dealing with the decentralized setting, particularly due to the
intractable form of the GDoF achieved under decentralized placement.
This intractability is circumvented by observing that the decentralized achievable GDoF
is bounded below by a centralized-like achievable GDoF, yet with a smaller coded-multicasting gain compared to the one achieved in a true centralized setting.
This key observation enables us to prove order-optimality in the decentralized setting.

In addition to the contributions highlighted above,
we derive several insights from the optimal GDoF characterization, which generalize former observations
obtained in special cases of the considered setting \cite{Zhang2015,Lampiris2017,Davoodi2016,Davoodi2018}.
Such insights, and how they relate to previous observations, can be found in Section \ref{sec:cen_results}.
As for the remainder of the paper, the organization is as follows.
Section \ref{sec:problem_setting} introduces the considered setting and problem.
Section \ref{sec:main_results} presents the two main results and related insights.
In Section \ref{sec:outer_bound}, we derive an outer bound which is employed in the following two sections to show order optimality.
In Section \ref{sec:centralized} and Section \ref{sec:decentralized}, we prove the two main results, the centralized setting result and the decentralized setting result respectively.
Section \ref{sec:conclusion} concludes the paper.
\newcounter{Theorem_Counter}
\newcounter{Proposition_Counter}
\newcounter{Lemma_Counter}
\newcounter{Remark_Counter}
\newcounter{Assumption_Counter}
\newcounter{Definition_Counter}
\section{Problem Setting} \label{sec:problem_setting}
Consider a MISO BC consisting of a $K$-antenna transmitter serving $K$ receivers (or users),
where users are equipped with a single-antenna each.
Users are indexed by the set $[K]\triangleq\{1,2,\dots,K\}$.
In a communication session, each user requests one file from a content library \mbox{$\mathcal{W} \triangleq \{W_{1},\ldots,W_{N}\}$} consisting of $N \geq K$ files, each of size $F$ bits.
We assume that the transmitter has access to the entire library (this applies to each radio head,
or remote antenna, in  physically distributed settings).
\begin{figure}[]
	\centering
	\includegraphics[width=0.6\textwidth]{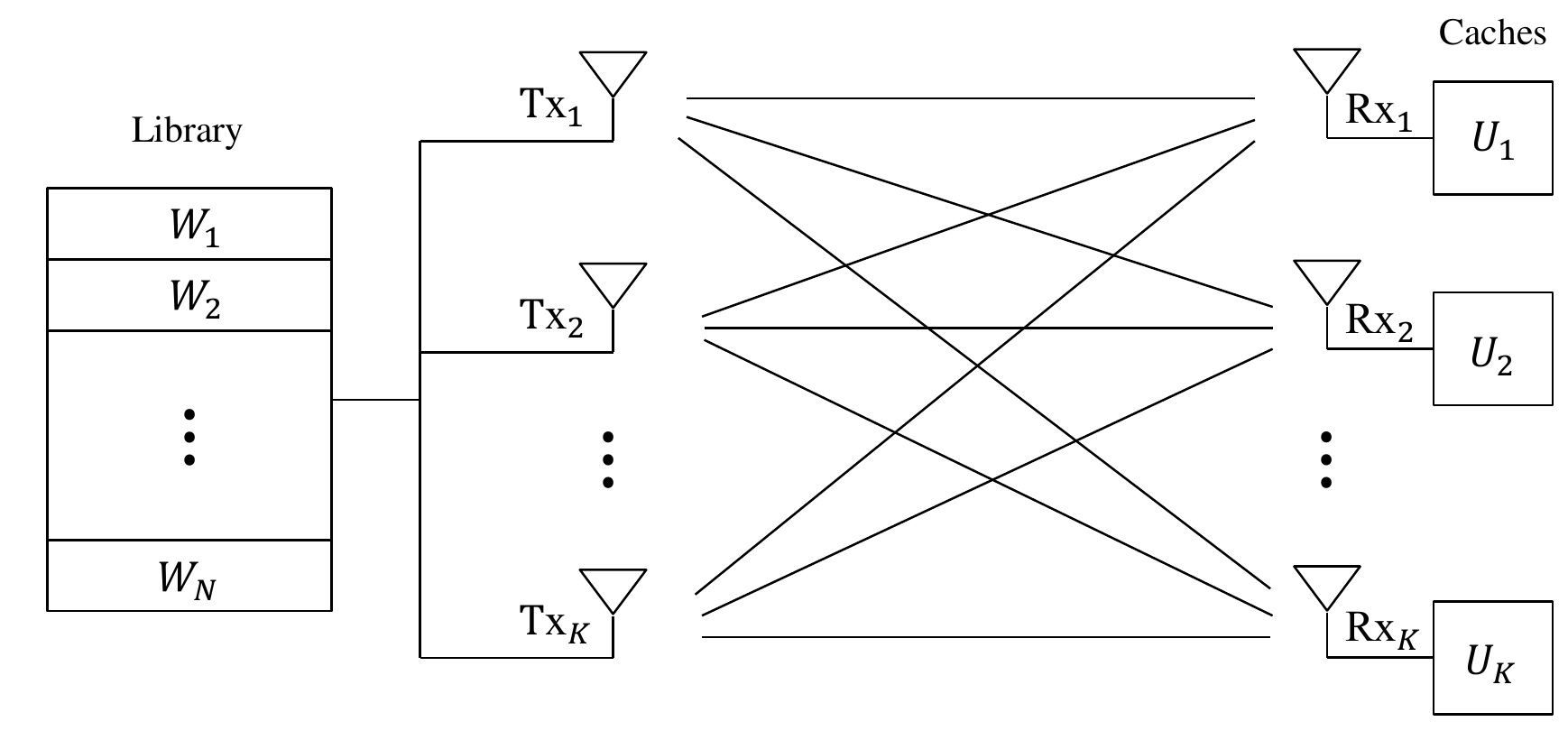}
	\caption{A wireless network in which a transmitter of $K$ antennas, $\text{Tx}_{1},\ldots,\text{Tx}_{K}$, serves
     $K$ single-antenna receivers, $\text{Rx}_{1},\ldots,\text{Rx}_{K}$. The transmitter has access to a library of $N$ files, while each receiver $\text{Rx}_{i}$ is equipped with a cache memory $U_i$.}
	\label{fig:fig_1}
\end{figure}

At the receiving end of the channel, each user $i$ is equipped with a cache memory $U_{i}$ of size $MF$ bits, where $M \in [0,N]$.
We define the \emph{normalized cache size} as
\begin{equation}
\mu \triangleq \frac{M}{N}
\end{equation}
which is interpreted as the fraction of the content library each user is able to store locally.
An illustration of the setup is given in Fig. \ref{fig:fig_1}.
It is readily seen that $\mu = 0$ reduces the setup to the classical MISO BC, while no communication needs to take place under
$\mu = 1$.
We refer to the $j$-th transmit antenna (or radio head) as the \emph{$j$-th transmitter} henceforth,
while \emph{transmitters} refers to the $K$ transmit antennas jointly.

The network  operates in two phases, \emph{a placement phase} and a \emph{delivery phase} \cite{Maddah-Ali2014}.
The placement phase takes place during the off-peak times before knowing the future demands of different users.
During this phase, the cache memories of the users are filled as
an arbitrary function of the $N$ files, where such function is denoted as $U_{i} = \phi_{i}(\mathcal{W})$.
The delivery phase takes place during peak times where each user requests one of the $N$ files. For example, user $i$ requests file $W_{d_{i}}$ for some $d_{i} \in [N]$, where $\mathbf{d} = (d_{1},\ldots,d_{K})$ is the tuple of all user demands.
Upon receiving the requests, each transmitter $j$ sends a codeword $X_{j}^{T} = X_{j}(1),\ldots,X_{j}(T)$ over $T\in \mathbb{N}$ uses of the \emph{physical channel}.
At the other end, each user $i$ receives the sequence $Y_{i}^{T} = Y_{i}(1),\ldots,Y_{i}(T)$, a noisy linear combination of the $K$ transmitted codewords.
The user then decodes for its requested file from $Y_{i}^{T}$ and the content of its own cache memory $U_{i}$.
This is described in more detail below.
\subsection{Physical Channel}
\label{subsec:physical channel}
The input-output relationship at the $t$-th use of the physical channel, $t \in [T]$, is modeled by
\begin{equation}
\label{eq:signal_model_1}
Y_{i}(t) = \sum_{j = 1}^{K} \sqrt{a_{ij}} G_{ij}(t)  X_{j}(t) + Z_{i}(t)
\end{equation}
where $Y_{i}(t) \in \mathbb{C}$ is the signal received by the $i$-th user, $X_{j}(t)\in \mathbb{C}$ is the $j$-th transmitter's normalized signal with power constraint $\E \left(|X_{j}(t)|^{2} \right) \leq 1$ and $Z_{i}(t) \sim \mathcal{N}_{\mathbb{C}}(0,1)$ is the normalized additive white Gaussian noise (AWGN), which is i.i.d. across all dimensions.
$a_{ij} \in \mathbb{R}_{+}$, $\forall j,i \in [K]$, captures the long-term constant gain of the \emph{link} between the $j$-th transmitter and the $i$-th receiver,
while $G_{ij}(t) \in \mathbb{C}$ is the corresponding time-varying  fading channel coefficient.
To avoid degenerate situations, we assume that the instantaneous value $|G_{ij}(t)|$ is bounded away from zero and infinity for all $i,j \in [K]$ and $t \in [T]$.
\subsubsection{GDoF Framework}
For any $i,j \in [K]$ and $i \neq j$, we refer to the link between transmitter $i$ and receiver $i$ as a \emph{direct-link},
while the link from transmitter $j$ to receiver $i$ is referred to as a \emph{cross-link}.
We consider a symmetric setup in which all direct-links (or cross-links) have similar long-term gains.
For GDoF purposes, we introduce the nominal SNR value $P \in \mathbb{R}_{+}$, simply referred to as the SNR henceforth.
Following the GDoF framework \cite{Etkin2008,Davoodi2017a},
channel gains are expressed in terms of the SNR as
\begin{equation}
\label{eq:SNR_INR}
a_{ii} = P \ \ \text{and} \ \ a_{ij} = P^{\alpha}, \ \forall i,j \in [K], \; i \neq j
\end{equation}
where the parameter $\alpha \geq 0$ quantifies the strength of cross-links.
The exponents of $P$ in \eqref{eq:SNR_INR}, i.e. $1$ and $\alpha$, are known as the channel strength parameters or levels.
The channel model in \eqref{eq:signal_model_1} is rewritten as
\begin{equation}
\label{eq:signal_model_2}
Y_{i}(t) = \sqrt{P} G_{ii}(t) X_{i}(t) +  \sum_{j = 1, j \neq i}^{K} \sqrt{P^{\alpha}} G_{ij}(t)  X_{j}(t) + Z_{i}(t)
\end{equation}
which is the model used throughout the paper.
The results in this paper are restricted to the regime $\alpha \in [0,1]$, i.e. scenarios in which the cross-link strength level is at most as strong as the direct-link strength level.
This is the most practically relevant regime, since each receiver associates with a
transmitter (i.e. radio head or remote antenna) from which it receives the strongest signal.
Moreover, as highlighted in \cite{Davoodi2018}, the regime $\alpha > 1$ poses new challenges and remains an open problem even for the classical MISO BC (with no caches) under partial CSIT.
\begin{remark}
As pointed out in \cite{Davoodi2017a}, the scaling of $P$ in the GDoF framework does not correspond to a
physical scaling of transmitting powers in a given channel (or network).
The correct interpretation is that each value of $P$ defines a new channel.
A class of channels parameterized by $\alpha$
belong together because the point-to-point capacity of any link (direct or cross) normalized by $\log(P)$ is approximately the same across all such channels belonging to the same class.
Hence, unlike the DoF framework, the GDoF framework preserves the diversity in link strengths as $P \rightarrow \infty$.
Moreover, DoF results are recovered from GDoF results by setting $\alpha = 1$, i.e, the special case in which all links are equally strong.
\end{remark}
\subsubsection{Partial CSIT}
\label{subsubsec:partial_CSIT}
Let $\mathcal{G} \triangleq \big\{ G_{ij}(t): i,j \in [K], \; t \in [T] \big\}$
be the set of  all channel coefficient variables.
Under partial CSIT, such channel coefficients may be represented as
\begin{equation}
G_{ij}(t) = \hat{G}_{ij}(t) + \sqrt{P^{-\beta}} \tilde{G}_{ij}(t)
\end{equation}
where $\hat{\mathcal{G}} \triangleq \big\{ \hat{G}_{ij}(t): i,j \in [K], \; t \in [T] \big\}$ are channel estimates, $\tilde{\mathcal{G}} \triangleq \big\{ \tilde{G}_{ij}(t): i,j \in [K], \; t \in [T] \big\}$ are estimation error terms and $\beta \in \mathbb{R}$ is a parameter capturing the CSIT quality level.
The channel knowledge available to the transmitters includes the coarse channel strength levels $\alpha$, the CSIT quality level $\beta$ and the estimates in $\hat{\mathcal{G}}$, but does not include the error terms in $\tilde{\mathcal{G}}$.

All variables in $\hat{\mathcal{G}}$ and $\tilde{\mathcal{G}}$ are subject to the bounded density assumption as explained in \cite{Davoodi2017a,Davoodi2018}.
The  difference between $\hat{\mathcal{G}}$ and $\tilde{\mathcal{G}}$, as pointed out earlier, is that the former is revealed to the transmitters while the latter is not.
Hence, given the estimates $\hat{\mathcal{G}}$, the variance of each channel coefficient in $\mathcal{G}$ behaves as $\sim P^{-\beta}$ and the peak of the probability density function behaves as $\sim \sqrt{P ^{\beta}}$.
Moreover, we assume throughout this work that $\beta \in [0,\alpha]$. In particular, $\beta = 0$ and $\beta = \alpha$ capture the two extremes where channel knowledge at the transmitters is absent and perfectly available, respectively \cite{Davoodi2018}.

Before we proceed, it is worth highlighting that channel state information at the receivers (CSIR) is assumed to be perfect.
Moreover, in a slight abuse of notation, we henceforth use $\hat{\mathcal{G}}$ to denote the entire channel knowledge available to the transmitters.
\subsection{Performance Measures}
Once transmitters are informed of the demands $\mathbf{d}$ in the delivery phase,
each transmitter $j$ generates a sequence of $T$ channel inputs $X_{j}^{T} = \psi_{j}^{(T)}(\mathcal{W},\mathbf{d},U_{1},\ldots,U_{K},\hat{\mathcal{G}})$, where
$\psi_{j}^{(T)}$ is an encoding function. Note that the availability of partial CSIT
is reflected in the argument $\hat{\mathcal{G}}$ of $\psi_{j}^{(T)}$.
Once the transmission is complete, each user $i$ maps its received signal, local cache content, user demands and perfect channel knowledge
to an estimate of the requested file $W_{d_{i}}$ denoted as $\hat{W}_{i} = \eta_{i}^{(T)}(Y_{i}^{T},U_{i},\mathbf{d},\mathcal{G})$, where $\eta_{i}$ is the decoding function.
The information theoretic limits of the system are studied by fixing $N,K,M,P$, and $\hat{\mathcal{G}}$, referred to as system parameters, while allowing $F$ and $T$ to grow arbitrarily large.

For fixed system parameters, a code which takes files of size $F$ bits and transmits codewords of block-length $T$ channel uses is defined as $\mathcal{C}^{(T)} \triangleq \big\{ \phi_{i},\psi_{i}^{(T)},\eta_{i}^{(T)}: i \in[K] \big\}$. It is evident that a code is characterized by its corresponding caching, encoding and decoding functions defined earlier.
The performance of a code is governed by its worst-case probability of error defined as
\begin{equation}
P_{e}^{(T)} \triangleq \max_{\mathcal{G}\mid\hat{\mathcal{G}}}  \max_{\mathbf{d} \in [N]^{K}} \max_{i \in [K]} \; \Prob \big( \hat{W}_{i} \neq W_{d_{i}} \big)
\end{equation}
which is taken over all possible users, for all possible demands, under all possible realizations of the channel coefficients given the available CSIT.
The (sum) rate of such code is defined as
\begin{equation}
\label{eq:rate_definition}
R \triangleq \frac{KF}{T}.
\end{equation}
For given system parameters, we say that the rate $R$ is achievable if there exists a coding scheme, consisting of a sequence of
codes $\left\{\mathcal{C}^{(T)} : T \in \mathbb{N} \right\}$ of  rate $R$ each, with a vanishing probability of error as the block-length grows arbitrarily large, i.e. $P_{e}^{(T)} \rightarrow 0$ as $T \rightarrow \infty$. Note that a strictly positive rate $R > 0$ requires $F\rightarrow \infty$
as $T\rightarrow \infty$.
The (sum) capacity $C$ is defined as the supremum of all achievable rates taken over all feasible coding schemes.
\subsubsection{GDoF}
By highlighting the dependency on the SNR $P$, it can be seen that each $P$ defines a new  channel (or network) with capacity $C(P)$.
The optimal (sum) GDoF is hence defined as
\begin{equation}
\mathsf{GDoF} \triangleq \lim_{P \rightarrow \infty}\frac{C(P)}{\log (P)}.
\end{equation}
Being an asymptotic (high-SNR) measure, it is well understood that the GDoF does not depend on $P$.
On the other hand, while fixing the number of users $K$, we often write $\mathsf{GDoF}(\mu,\alpha,\beta)$ to highlight the dependency on the system parameters $\mu$, $\alpha$ and $\beta$.
In particular, it turns out that our GDoF characterization is expressed in terms of the normalized cache size $\mu=M/N$ instead of the exact $N$ and $M$, and the cross-link strength level $\alpha$ and partial CSIT level $\beta$ instead of the entire CSIT $\hat{\mathcal{G}}$.
These observations are consistent with existing DoF results for cache-aided networks on one hand \cite{Zhang2017,Naderializadeh2017,Hachem2018}, and GDoF studies in classical networks under finite precision and partial CSIT on the other hand \cite{Davoodi2017a,Davoodi2018}.
\subsubsection{Generalized Normalized Delivery Time}
Instead of working directly with the $\mathsf{GDoF}$, it is easier to derive the results in terms of a function of the reciprocal\footnote{This has been observed when dealing with the DoF in many works including \cite{Maddah-Ali2015,Zhang2017,Hachem2018,Xu2017}.} $1/\mathsf{GDoF}$.
Hence, we introduce the \emph{generalized normalized delivery time} (GNDT), where the optimal GNDT is defined as
\begin{equation} \label{eq:link_deliverytime_gdof}
 \mathsf{GNDT}(\mu,\alpha,\beta)  \triangleq \frac{K}{ \mathsf{GDoF}(\mu,\alpha,\beta)}.
\end{equation}
The GNDT (or the \emph{delivery time} as we refer to it throughout the paper) is measured in \emph{time-slot}. One time-slot is the optimal amount of time required to communicate a single file to a single user over a direct-link (with strength level $1$) under no caching and no interference as $P \rightarrow \infty$.
In particular, since a single user direct-link with no interference and no caching has a capacity of $\log (P) + o\big(\log (P) \big)$, i.e. $\mathsf{GDoF} = 1$, it is readily seen that $\mathsf{GNDT} = 1$ time-slot for such setting.
For any given $\mu$, $\alpha$ and $\beta$, we say that the delivery time  $\mathsf{GNDT}'(\mu,\alpha,\beta)$ is achievable if
$\mathsf{GNDT}'(\mu,\alpha,\beta) \geq \mathsf{GNDT}(\mu,\alpha,\beta)$.

The GNDT generalizes the \emph{normalized delivery time} (NDT) metric in \cite{Sengupta2017} to suit the GDoF framework.
Hence, it is not surprising to observe that the GNDT-GDoF relationship resembles (and generalizes) the NDT-DoF relationship.
Moreover, it is readily seen from \eqref{eq:link_deliverytime_gdof} that the GDoF can be interpreted as the capacity in files per time-slot.
Before we proceed, we remark that in this paper, as in \cite{Maddah-Ali2014,Maddah-Ali2015,Naderializadeh2017,Zhang2017,Hachem2018,Lampiris2017,Maddah-Ali2015a,Sengupta2017,Xu2017},
we adopt a worst-case definition of performance measures with respect to user requests.
As a result, it is always assumed that each user requests a different file.
\subsection{Centralized Placement vs. Decentralized Placement}
Although the placement phase does not depend on the actual user demands $\mathbf{d}$ in the delivery phase,
placement strategies may still depend on the identity and number of active users during the delivery phase.
Such coordination in the placement phase is known as centralized placement.
Since the identity, or even the number, of active users may not be known several hours before the delivery phase takes place,
it is also important to consider strategies in which placement is not allowed to depend on such information.
This lack of coordination is known as decentralized placement \cite{Maddah-Ali2015a}.
Decentralization during the placement phase can be realized by allowing
randomized placement schemes. For instance, each user $i$ independently
draws a caching function $\phi_{i}(\mathcal{W};D)$ from
an ensemble of randomized caching functions parameterized by an arbitrary random variable
$D$, independent of $i$ and $K$.
\section{Main Results and Insights} \label{sec:main_results}
The main results of this paper are: 1) the GDoF characterization of the symmetric cache-aided MISO BC under partial CSIT, described in Section \ref{sec:problem_setting}, to within a constant multiplicative gap, and 2) showing that such GDoF characterization is robust to decentralization.
We start by presenting the first result and deriving useful insights assuming a centralized setting, then we extend to the decentralized setting.
\subsection{Centralized placement}
\label{sec:cen_results}
In order to state the GDoF result, we define the centralized GNDT function $\mathsf{GNDT}_{\mathrm{C}}(\mu,\alpha,\beta)$, where
\begin{equation} \label{eq:T_cen}
\mathsf{GNDT}_{\mathrm{C}}(\mu,\alpha,\beta) \triangleq \frac{K(1-\mu)}{K(1-(\alpha-\beta))+(1+K\mu)(\alpha-\beta)}
\end{equation}
for any $\alpha \in [0,1]$, $\beta \in [0,\alpha]$ and $\mu \in \{0,\frac{1}{K},\frac{2}{K}, \dots, \frac{K-1}{K},1 \}$, and the lower convex envelope of these points for all other $\mu \in [0,1]$.
\begin{theorem} \label{th_cen}
For the symmetric cache-aided MISO BC under partial CSIT described in Section \ref{sec:problem_setting},
under centralized placement we achieve the GDoF given by
\begin{equation} \label{eq:GDoF_cen}
\mathsf{GDoF}_{\mathrm{C}}(\mu,\alpha,\beta)=\frac{K}{\mathsf{GNDT}_{\mathrm{C}}(\mu,\alpha,\beta)}.
\end{equation}
Moreover, the achievable GDoF in \eqref{eq:GDoF_cen} satisfies
\begin{equation} \label{eq:converse_cen}
\mathsf{GDoF}_{\mathrm{C}}(\mu,\alpha,\beta) \leq \mathsf{GDoF}(\mu,\alpha,\beta) \leq 12 \cdot\mathsf{GDoF}_{\mathrm{C}}(\mu,\alpha,\beta).
\end{equation}
\end{theorem}
The proof of Theorem \ref{th_cen} is presented in Section \ref{sec:centralized}.
As in \cite{Maddah-Ali2015,Hachem2018}, the somewhat loose multiplicative gap of $12$ in Theorem \ref{th_cen}
is due to the analytical bounding techniques used in the converse.
Numerical simulations suggest that such factor is no more than $3.5$ for $K \leq 100$ and $N \leq 500$.

To gain some insights into the GDoF characterized in Theorem \ref{th_cen}, we restrict the following discussion to $\mu \in \{0,\frac{1}{K},\frac{2}{K}, \dots, \frac{K-1}{K} \}$, for which the achievable GDoF in \eqref{eq:GDoF_cen} is expressed as
\begin{align} \label{eq:GDoF_cen_2}
\mathsf{GDoF}_{\mathrm{C}}(\mu,\alpha,\beta) & =(1-(\alpha-\beta))\frac{K}{1-\mu}+(\alpha-\beta)\frac{1+K\mu}{1-\mu}.
\end{align}
It is easily seen that $\mathsf{GDoF}_{\mathrm{C}}(\mu,\alpha,\beta)$ in \eqref{eq:GDoF_cen_2} reduces to its classical counterpart in \cite{Davoodi2018} under $\mu = 0$, i.e. where no caches are available.
In this case, the multiplicative factor of $12$ can be reduced to $1$.
However, more significantly, the form taken by the GDoF in \eqref{eq:GDoF_cen_2}, for any $\mu$ (in the set above), is analogous
to the form of the classical GDoF in \cite{Davoodi2018}.
This is explained in more details next,
where we use the terminology of signal power levels measured in terms of the exponent of $P$ \cite{Avestimehr2015}.
We start by looking at specialized cases from which we build our way towards the general case.
\subsubsection{DoF Under Partial CSIT}
Recall that DoF characterization under partial CSIT is obtained by setting $\alpha = 1$.
Defining $\mathsf{DoF}_{\mathrm{C}}(\mu,\beta) \triangleq \mathsf{GDoF}_{\mathrm{C}}(\mu,1,\beta)$ and applying such specialization to \eqref{eq:GDoF_cen_2}, we obtain
\begin{align}
\label{eq:DoF_cen_2}
\mathsf{DoF}_{\mathrm{C}}(\mu,\beta) & =\beta\frac{K}{1-\mu}+(1-\beta)\frac{1+K\mu}{1-\mu}.
\end{align}
Under perfect CSIT ($\beta = 1$), zero-forcing over the physical channel enables a spatial multiplexing gain of $K$.
By incorporating caches into the picture, we obtain a further local caching gain of $\frac{1}{1-\mu}$,
which is the only relevant caching gain here as zero-forcing creates parallel (non-interfering) single-user links.
Under the other extreme, i.e. finite precision CSIT ($\beta = 0$),
all spatial multiplexing gains in the physical channel are lost and
the DoF collapses to the one obtained in the original setting with a shared
link \cite{Maddah-Ali2014}.
In this case, the network relies on the local caching gain of $\frac{1}{1-\mu}$ and the global caching gain of $1+K\mu$, where the latter is enabled by creating coded-multicasting opportunities.

It is readily seen that finite precision CSIT is as (un)useful as no CSIT from a DoF perspective\footnote{It is implicitly understood that such statements hold in an order-optimal sense. This applies to all similar observations herein.}. This is reminiscent of the DoF collapse
in the classical MISO BC \cite{Davoodi2016}.
Moreover, it is worth noting that since the DoF of the cache-aided MISO BC is an upper bound for the DoF of cache-aided interference networks,
this collapse under finite precision CSIT also holds for the networks in \cite{Naderializadeh2017,Hachem2018,Xu2017}.

For partial CSIT ($0 < \beta < 1$), the DoF takes the form
$\beta\mathsf{DoF}_{\mathrm{C}}(\mu,1) + (1-\beta) \mathsf{DoF}_{\mathrm{C}}(\mu,0)$,
laying on the line connecting the two extremes.
In this case, partial CSIT of level $\beta$ allows (power-controlled) zero-forcing transmission in the bottom $\beta$ signal power levels without leaking any interference above the noise floor at undesired users.
This utilization of only a fraction of power levels yields the factor $\beta$ in the DoF.
The remaining signal power levels are used for a shared-link-type transmission requiring no CSIT.
In particular, this transmission sees interference from the zero-forcing layer, hence is left with the top $(1-\beta)$ power levels as reflected
in the DoF.
Since all users can decode (and remove) all codewords in the shared link layer without influencing its achievable DoF, the zero-forcing layer remains unaffected.
To facilitate such partitioned transmission, messages (or files) are split into private and common parts delivered through the zero-forcing and shared link layers, respectively.

The scheme described above expands upon, and inherits the main features of, the rate-splitting scheme\footnote{Also known as signal space partitioning \cite{Yuan2016}.} used for the
classical MISO BC with partial CSIT (alongside other networks) \cite{Yang2013,Joudeh2016,Hao2017,Yuan2016,Davoodi2018,Davoodi2017}.
Hence, it is not surprising to see that the cache-aided DoF takes the same
weighted-sum form of the classical DoF in\cite{Joudeh2016}, recovered from the above by setting $\mu = 0$.
\subsubsection{GDoF Under Finite Precision CSIT}
\label{subsubsec:GDoF_no_CSIT}
This is recovered from \eqref{eq:GDoF_cen_2} by setting $\beta = 0$ and corresponds to the achievable GDoF in \cite{Lampiris2017}.
It is easily checked that the GDoF in this case takes the form of the DoF in \eqref{eq:GDoF_cen_2}, after replacing $\beta$ with $1 - \alpha$.
This is inline with the observation that DoF results under partial CSIT translate to
GDoF results under finite precision CSIT \cite{Yuan2016}.
This also highlights that unlike the DoF metric, the GDoF metric captures spatial multiplexing gains under finite precision (or even absent) CSIT.
Such multiplexing gains, however, are achieved by exploiting the signal power levels only.
\subsubsection{The General Case}
For arbitrary levels of $\beta$ and $\alpha$, the insights derived in \cite{Davoodi2018} for the GDoF of the classical MISO BC
extend to the cache-aided counterpart.
In particular, the cross-link strength level $\alpha$ and the CSIT quality level $\beta$ equally counter each other and hence only their difference $(\alpha-\beta)$ matters.
The bottom $1 - (\alpha-\beta)$ power levels are reserved for parallel-link-type transmission through zero-forcing and power control, while the shared-link-type transmission rises above, essentially occupying the top $(\alpha-\beta)$ power levels.
Therefore, it is readily seen that as $(\alpha-\beta)$ increases, the network starts relying more on the global caching gain and less on spatial multiplexing gains as reflected in \eqref{eq:GDoF_cen_2}.
\subsection{Decentralized placement}
In this part we consider the decentralized setting where centrally coordinated placement is not allowed during the placement phase.
Before we state the following result, we define the decentralized GNDT function $\mathsf{GNDT}_{\mathrm{D}}(\mu,\alpha,\beta)$, where
\begin{equation} \label{eq:T_decen}
\mathsf{GNDT}_{\mathrm{D}}(\mu,\alpha,\beta) \triangleq
{K  \sum_{m=0}^{K-1}{\frac{\binom{K-1}{m} \mu^m\left(1-\mu\right)^{K-m}}{K(1-(\alpha-\beta))+(1+m)(\alpha-\beta)}}}
\end{equation}
for any $\alpha \in [0,1]$, $\beta \in[ 0,\alpha]$ and $\mu \in [0,1]$.
\begin{theorem} \label{th_decen}
For the symmetric cache-aided MISO BC under partial CSIT described in Section \ref{sec:problem_setting},
under decentralized placement we achieve the GDoF given by
	\begin{equation} \label{eq:GDoF_decen}
	\mathsf{GDoF}_{\mathrm{D}}(\mu,\alpha,\beta)= \frac{K}{\mathsf{GNDT}_{\mathrm{D}}(\mu,\alpha,\beta)}.
	\end{equation}
Moreover, the achievable GDoF in \eqref{eq:GDoF_decen} satisfies
	\begin{equation} \label{eq:converse_decen}
	\mathsf{GDoF}_{\mathrm{D}}(\mu,\alpha,\beta) \leq \mathsf{GDoF}(\mu,\alpha,\beta) \leq 12 \cdot\mathsf{GDoF}_{\mathrm{D}}(\mu,\alpha,\beta).
	\end{equation}
	
\end{theorem}
The proof of Theorem \ref{th_decen} is presented in Section \ref{sec:decentralized}.
The most significant consequence of Theorem \ref{th_decen} is that centralized placement leads to at most a constant-factor improvement of the GDoF over decentralized placement.
Through a straightforward inspection, this constant-factor improvement is bounded above by
$\mathsf{GDoF}_{\mathrm{C}}(\mu,\alpha,\beta) \leq 12 \cdot \mathsf{GDoF}_{\mathrm{D}}(\mu,\alpha,\beta)$, obtained from \eqref{eq:converse_cen} and \eqref{eq:converse_decen}.
In Section \ref{subsec:gap_cen_decen},
this multiplicative gap between the centralized GDoF and decentralized GDoF is tightened
to $1.5$.

In Section \ref{sec:proof_th_opt_decen}, we show that an upper bound on $\mathsf{GNDT}_{\mathrm{D}}(\mu,\alpha,\beta)$ takes the form of the centralized delivery time in \eqref{eq:T_cen}, yet with a lower coded-multicasting gain.
It follows that the insights that follow Theorem \ref{th_cen}, derived in the light of the centralized achievable GDoF, extend to the decentralized setting.
\section{Outer Bound}
\label{sec:outer_bound}
In this section, we obtain an outer bound (upper bound) for the GDoF.
Since it is more convenient to work with the GNDT in \eqref{eq:link_deliverytime_gdof}, the
outer bound is derived in terms of a lower bound on $ \mathsf{GNDT}(\mu,\alpha,\beta)$.
\begin{theorem} \label{th_lower_bound}
For the symmetric cache-aided MISO BC under partial CSIT described in Section \ref{sec:problem_setting}, a lower bound on the optimal GNDT is given by
	\begin{equation} \label{eq:cut_set_bound}
	\mathsf{GNDT}(\mu,\alpha,\beta) \geq \max_{s \in \{1,2,\dots,K\}} \mathsf{GNDT}^{\mathrm{lb}}_s(\mu,\alpha,\beta),
	\end{equation}
	where $\mathsf{GNDT}^{\mathrm{lb}}_s(\mu,\alpha,\beta)$ is defined as\footnote{For any $x \in \mathbb{R}$, we define $(x)^{+} \triangleq \max\{0,x\}$.}
	\begin{equation} \label{eq:cut_set_bound_s}
	\mathsf{GNDT}^{\mathrm{lb}}_s(\mu,\alpha,\beta) \triangleq \left(\frac{s}{1+(s-1)(1-(\alpha-\beta))}\left(1-\frac{M}{\left \lfloor{\frac{N}{s}} \right  \rfloor}\right) \right)^+.
	\end{equation}
\end{theorem}
In the above, for any subset of $s \leq K$ users, the corresponding
$\mathsf{GNDT}^{\mathrm{lb}}_s(\mu,\alpha,\beta)$ in \eqref{eq:cut_set_bound_s}
is a lower bound on the optimal delivery time $\mathsf{GNDT}(\mu,\alpha,\beta)$.
It follows that the tightest of such lower bounds is obtained by maximizing
$\mathsf{GNDT}^{\mathrm{lb}}_s(\mu,\alpha,\beta)$ over $s$.
We also observe that $\mathsf{GNDT}^{\mathrm{lb}}_s(\mu,\alpha,\beta)$ depends on the parameters of the physical channel through the difference $(\alpha-\beta)$.
In particular, for a fixed number of users $s$, library size $N$ and cache size $M$, $\mathsf{GNDT}^{\mathrm{lb}}_s(\mu,\alpha,\beta)$
decreases when $(\alpha-\beta)$ decreases.
This is intuitively explained by the fact that decreasing $(\alpha-\beta)$
corresponds to higher (relative) CSIT quality, enabling larger spatial multiplexing gains which in turn
reduce the delivery time.

From Theorem  \ref{th_lower_bound} and \eqref{eq:link_deliverytime_gdof}, it is easily seen that an upper bound for the GDoF is given by
\begin{equation} \label{eq:cut_set_bound_GDoF}
\mathsf{GDoF}(\mu,\alpha,\beta) \leq \min_{s \in \{1,2,\dots,K\}} \frac{K}{\mathsf{GNDT}^{\mathrm{lb}}_s(\mu,\alpha,\beta)}.	
\end{equation}
The outer bound in Theorem \ref{th_lower_bound} is employed to prove the converse parts of Theorem \ref{th_cen} and Theorem \ref{th_decen} in the following sections.
In the remainder of this section, we present a proof for Theorem \ref{th_lower_bound}.
The proof relies on  two main ingredients summarized as follows.
\begin{enumerate}[(a)]
\item A lower bound on $\mathsf{GNDT}(\mu,\alpha,\beta)$ is obtained by considering a subset of $s \leq K$ users and a multi-demand communication, in which each user requests multiple distinct files.
\item Each cache memory is replaced with a parallel side link of capacity that can convey the information content of the cache to the user by the end of the multi-demand communication.
    By bounding the GDoF of this new channel, we bound the delivery time of the multi-demand communication.
\end{enumerate}
Similarities and differences between this proof and previous works are discussed at the end of this section.
\subsection{Multi-Demand Communication}
Consider a subset of $s \leq K$ users and a multi-demand communication over the cache-aided channel,
in which each user requests a set of $\left \lfloor{\frac{N}{s}}\right \rfloor$ distinct files and no file is requested by two
different users.
We denote the $\left \lfloor{\frac{N}{s}}\right \rfloor$  files requested by user $i$ as
$W_{d_{i}^{1}},\ldots,W_{d_{i}^{\left \lfloor{{N}/{s}}\right \rfloor}}$.
By the end of the communication, each user is able to recover the $\left\lfloor{\frac{N}{s}}\right\rfloor$ requested
files from the received signals and the local cache content.
The optimal delivery time for this multi-demand communication is denoted by $\mathsf{GNDT}_{\mathrm{md}}$,
which is also defined in the worst-case sense, i.e. for the worst-case amongst all possible multi-demands of $\left \lfloor{\frac{N}{s}}\right \rfloor$ files.
It is readily seen that $\mathsf{GNDT}_{\mathrm{md}}$ satisfies
\begin{equation}
\label{eq:delivery_time_multi_session}
\mathsf{GNDT}_{\mathrm{md}} \leq \left\lfloor{\frac{N}{s}}\right\rfloor\mathsf{GNDT}(\mu,\alpha,\beta)
\end{equation}
since we are ignoring $K-s$ users and it is always feasible to treat each demand of $s$ files separately in a consecutive manner.
Next, we transfer to an equivalent setup with no caches.
\subsection{Cache Replacement and Delivery Time Lower Bound}
Now consider a new MISO BC consisting of the same $K$ transmitters, with access to the same library of $N$ files, and the
$s \leq K$ users served in the multi-demand communication above.
However, users in this new channel are not equipped with caches.
Alternatively, communication is carried out over two parallel sub-channels.
The input-output relationship is given by
\begin{align}
\label{eq:signal_model_parallel_1}
Y_{i}(t) &= \sqrt{P} G_{ii}(t) X_{i}(t) +  \sum_{j = 1, j \neq i}^{K} \sqrt{P^{\alpha}} G_{ij}(t)  X_{j}(t) + Z_{i}(t) \\
\label{eq:signal_model_parallel_2}
B_{i}(t) &= \sqrt{P^{\gamma}}  A_{i}(t)  + C_{i}(t)
\end{align}
where \eqref{eq:signal_model_parallel_1} and \eqref{eq:signal_model_parallel_2} describe the first and second sub-channels, respectively.
All physical properties of \eqref{eq:signal_model_2}, described in Section \ref{subsec:physical channel}, are inherited by the first sub-channel in \eqref{eq:signal_model_parallel_1}.
For the second sub-channel, $A_{i}(t) \in \mathbb{C}$ is the signal transmitted to the $i$-th user with a
power constraint $\E \left(|A_{i}(t)|^{2} \right) \leq 1$, $B_{i}(t) \in \mathbb{C}$ is the signal received by the $i$-th user
and $C_{i}(t) \sim \mathcal{N}_{\mathbb{C}}(0,1)$ is the i.i.d. AWGN.
Each link in the second sub-channel remains constant over $t$ and has channel strength level $\gamma \geq 0$,
hence supports a transmission at rate $\gamma\log (P) + o\big(\log (P) \big)$ without influencing the rate over the first sub-channel.
Equivalently, $\gamma$ is the GDoF (or capacity in files per time-slot) of each individual link in the second sub-channel.

In this new MISO BC with parallel sub-channels, each user $i$ requests the
same $\left \lfloor{\frac{N}{s}}\right \rfloor$
files requested by the corresponding user in the multi-demand communication, i.e. $W_{d_{i}^{1}},\ldots,W_{d_{i}^{\left \lfloor{{N}/{s}}\right \rfloor}}$.
Each transmitter $j$ then generates the codewords $X_{j}^{n}$ and $A_{j}^{n}$, sent
over $n \in \mathbb{N}$ channel uses through the sub-channels in \eqref{eq:signal_model_parallel_1} and \eqref{eq:signal_model_parallel_2} respectively.
By the end of the communication, user $i$ recovers the $\left \lfloor{{N}/{s}}\right \rfloor$  requested files from the signals $Y_{i}^{n}$ and $B_{i}^{n}$, received through the sub-channels in \eqref{eq:signal_model_parallel_1} and \eqref{eq:signal_model_parallel_2} respectively.
The optimal (sum) GDoF of this new MISO BC, denoted by $\mathsf{GDoF}_{\mathrm{P}}(\alpha,\beta,\gamma) $,  is bounded above as follows.
\begin{lemma} \label{lemma:AIS}
For the $s$-user MISO BC, consisting of two parallel sub-channels, described in \eqref{eq:signal_model_parallel_1} and \eqref{eq:signal_model_parallel_2}, the optimal (sum) GDoF is bounded above as
	\begin{equation}
	\label{eq:GDoF_parallel}
	\mathsf{GDoF}_{\mathrm{P}}(\alpha,\beta,\gamma) \leq  (\alpha - \beta) + s\big( 1 - (\alpha - \beta) \big) + s\gamma.
	\end{equation}
\end{lemma}
It is evident that the bound on $\mathsf{GDoF}_{\mathrm{P}}(\alpha,\beta,\gamma)$ in \eqref{eq:GDoF_parallel} depends
 on $\alpha$ and $\beta$ through their difference $(\alpha-\beta)$. For the extreme case of $(\alpha-\beta) = 0$,
 the parallel MISO BC enjoys full spatial multiplexing gains over the first sub-channel.
 On the other hand, for the other extreme of $(\alpha-\beta) = 1$, all spatial multiplexing gains are annihilated and the GDoF of the first sub-channel collapses to $1$. Note that the contribution from the second sub-channel is unaffected since it consists of non-interfering links.
The proof of Lemma \ref{lemma:AIS} is relegated to Appendix \ref{sec:proof_AIS_lemma}.
Next, we argue that by setting $\gamma$ such that
\begin{equation}
\label{eq:eq_d0}
\gamma \cdot \mathsf{GNDT}_{\mathrm{md}} = M
\end{equation}
the corresponding optimal delivery time of the new channel is a lower bound on the optimal total delivery time of the cache-aided multi-demand communication, i.e.
\begin{equation}
\label{eq:delivery_time_multi_session_lb}
\frac{s \left \lfloor{\frac{N}{s}}\right \rfloor }{\mathsf{GDoF}_{\mathrm{P}}(\alpha,\beta,\gamma)} \leq  \mathsf{GNDT}_{\mathrm{md}}.
\end{equation}
This follows by observing that \eqref{eq:eq_d0} guarantees that for each user $i$,
the content of the cache $U_{i}$ in the original channel
can be delivered over the second sub-channel in \eqref{eq:signal_model_parallel_2} using at most $\mathsf{GNDT}_{\mathrm{md}}$ time-slots.
Since this does not influence the GDoF achieved over the first sub-channel in \eqref{eq:signal_model_parallel_1},
any placement and delivery strategy implemented for the cache-aided multi-demand communication is feasible in the new channel
and will take at most $\mathsf{GNDT}_{\mathrm{md}}$ time-slots.
We proceed while assuming that \eqref{eq:eq_d0} holds.

By combining \eqref{eq:delivery_time_multi_session_lb} with Lemma \ref{lemma:AIS} and \eqref{eq:eq_d0}, followed by invoking \eqref{eq:delivery_time_multi_session}, we obtain
\begin{align}
\left \lfloor{\frac{N}{s}}\right \rfloor s & \leq \mathsf{GNDT}_{\mathrm{md}} \big(1+(s-1)(1-(\alpha-\beta))+s\gamma \big) \\
& = \mathsf{GNDT}_{\mathrm{md}} \big(1+(s-1)(1-(\alpha-\beta)) \big) + sM \\
& \leq \mathsf{GNDT}(\mu,\alpha,\beta) \left \lfloor{\frac{N}{s}}\right \rfloor \big(1+(s-1)(1-(\alpha-\beta)) \big) + sM.
\end{align}
After some rearrangement and by considering that the delivery time is non-negative, we obtain
\begin{equation}
\label{eq:eq_Tsalpha_lb}
\mathsf{GNDT}(\mu,\alpha,\beta) \geq \left(\frac{s}{1+(s-1)(1-(\alpha-\beta))}\left(1-\frac{M}{\left \lfloor{\frac{N}{s}} \right  \rfloor}\right) \right)^+.
\end{equation}
The lower bound in \eqref{eq:eq_Tsalpha_lb} is further tightened by maximizing over all possible sizes of user subsets, i.e. $s \in [K]$, from which the result in \eqref{eq:cut_set_bound} directly follows.
\subsection{Insights and Relation to Prior Works}
The multi-demand communication to a subset of users corresponds to the cut-set-based bound in \cite{Maddah-Ali2014}, while the cache replacement is inspired by \cite{Zhang2017}.
However, it is worthwhile highlighting that bounding the DoF under partial current and perfect delayed CSIT and side links (after cache replacement) in \cite{Zhang2017} is very different from bounding the GDoF under only partial current CSIT and side links in Lemma \ref{lemma:AIS}.
In particular, the DoF upper bound in \cite{Zhang2017} follows the footsteps of \cite{Chen2016},
and is essentially based on a genie-aided argument.
Such argument does not work for the DoF/GDoF with only partial current CSIT and is known to give a loose bound in general.
The proof of Lemma \ref{lemma:AIS} is hence based on the outer bounds in
 \cite{Davoodi2016,Davoodi2017a,Davoodi2018}, which rely on the aligned image sets approach under channel uncertainty.

It is also worthwhile highlighting that the GDoF upper bound in Lemma \ref{lemma:AIS} is achievable through separate coding over the two sub-channels, i.e. there are no synergistic gains to be exploited through joint coding.
This comes in contrast to the setting in \cite{Zhang2017}, where jointly coding over the parallel sub-channels (after cache replacement)
can strictly outperform separate coding.
The influence of this synergy (or the lack of it) is clear when we revert back to the cache-aided channels.
In particular, we saw in Theorem \ref{th_cen} that the considered cache-aided MISO BC collapses to the shared-link setting in \cite{Maddah-Ali2014} when $(\alpha - \beta) = 1$.
However, even when current CSIT is completely absent in \cite{Zhang2017}, the synergy between caches and delayed CSIT
leads to an improved performance compared to the shared-link setting.
\section{Centralized Placement}
\label{sec:centralized}
In this section, we treat the centralized setting and prove Theorem \ref{th_cen}.
We start with the achievability and then we prove order-optimality using the outer bound in Theorem \ref{th_lower_bound}.
\subsection{Achievability scheme}
Here we present a centralized scheme which achieves the delivery time given by $\mathsf{GNDT}_{\mathrm{C}}(\mu,\alpha,\beta)$ in \eqref{eq:T_cen},
and hence the GDoF given by $\mathsf{GDoF}_{\mathrm{C}}(\mu,\alpha,\beta)$ in Theorem \ref{th_cen}.
This scheme builds upon and generalizes the one proposed for the cache-aided MISO BC in \cite{Lampiris2017}.
The key difference is that the scheme in \cite{Lampiris2017} is tuned to a special case in which only finite precision CSIT (i.e. $\beta = 0$) is available, while the one proposed here bridges the gap by considering all relevant levels of partial CSIT, i.e. $\beta \in [0,\alpha]$.

A key ingredient of the achievability scheme is the transmission of common and private codewords during the delivery phase.
We start by treating this physical-layer aspect through the following result.
\begin{lemma}
\label{lemma:private_common_GDoF}
Consider the $K$-user MISO BC with signal model given by \eqref{eq:signal_model_2} and properties described
in Section \ref{subsec:physical channel}.
Further assume that the transmitter has a common message $W^{(\mathrm{c})}$, intended to all user, and private messages
$W_1^{(\mathrm{p})},\ldots,W_K^{(\mathrm{p})}$, where $W_i^{(\mathrm{p})}$ is intended only to user $i$.
We achieve the GDoF
\begin{align}
\label{eq:GDoF_common}
\mathsf{GDoF}^{(\mathrm{c})} & = (\alpha - \beta)  \\
\label{eq:GDoF_private}
\mathsf{GDoF}_{i}^{(\mathrm{p})} & = 1 -  (\alpha - \beta), \ \forall i \in [K]
\end{align}
where $\mathsf{GDoF}^{(\mathrm{c})} $  is the GDoF achieved by the common message and
$\mathsf{GDoF}_{i}^{(\mathrm{p})}$ is the GDoF achieved by the
$i$-th private message.
\end{lemma}
The GDoF in  \eqref{eq:GDoF_common} and \eqref{eq:GDoF_private} is achieved using signal space partitioning \cite{Davoodi2018,Yuan2016}.
Using the terminology of signal power levels to explain this partitioning,
the upper $(\alpha - \beta)$ power levels are occupied by the common message while the bottom  $1 -  (\alpha - \beta)$ power levels are reserved for the private messages.
Note that the transmission of the common message requires no CSIT, while the transmission of the private messages is carried out using zero-forcing and power control, and hence may rely on the available partial CSIT.
Therefore, in the extreme case of $(\alpha - \beta) = 1$ (i.e. finite precision CSIT and equal strength paths), spatial multiplexing gains achieved through zero-forcing and power control collapse and the corresponding private messages will have a GDoF of zero.
The full proof of Lemma \ref{lemma:private_common_GDoF} is relegated to Appendix \ref{sec:proof_private_common_GDoF}.

In the following, we focus on $\mu \in \{\frac{1}{K},\frac{2}{K}, \dots, \frac{K-1}{K} \}$, such that $K\mu$ is an integer.
For $\mu = 0$, no caching is possible and the GDoF-optimal transmission strategy is given in \cite{Davoodi2018}.
For the other extreme of $\mu=1$, we have $\mathsf{GNDT}_{\mathrm{C}}(1,\alpha,\beta)=0$ as each user can store the entire library.
For the remaining $\mu$, where $K\mu$ is not necessarily an integer, $\mathsf{GNDT}_{\mathrm{C}}(\mu,\alpha,\beta)$ is obtained by
memory-sharing over the schemes corresponding to $\mu \in \big\{0,\frac{1}{K},\frac{2}{K}, \dots, \frac{K-1}{K}, 1 \big\}$, as pointed out in \cite{Maddah-Ali2014}.
\subsubsection{Placement phase}
The placement is analogous to \cite{Maddah-Ali2014}
and does not depend on the parameters specific the considered channel, e.g. transmitting antennas, $\alpha$ and $\beta$.
We use $m_{\mathrm{C}} \triangleq \mu K$ for notational briefness and to facilitate reusing some parts in the following section for the decentralized case.
Let $\Omega=\{ \mathcal{T} \subseteq [K]: |\mathcal{T}|= m_{\mathrm{C}}  \}$ be the family of all subsets of users with cardinality $m_{\mathrm{C}}$.
Each file $W_l \in \mathcal{W}$ is split into ${K \choose m_{\mathrm{C}}}$ non overlapping, equal size, subfiles
labeled as $W_{l,\mathcal{T}}$, for all $\mathcal{T} \in \Omega$, where each subfile consists of $F/{K \choose m_{\mathrm{C}} }$ bits.
User $i$ caches all the subfiles $W_{l, \mathcal{T}}$ such that $i \in \mathcal{T}$ and $l \in [N]$.
Hence, the corresponding cache memory is filled as $U_i=\{W_{l,\mathcal{T}}: \: \mathcal{T} \in \Omega, \: i \in \mathcal{T}, \: l \in [N]\}$.
Each user stores $N \binom{K-1}{m_{\mathrm{C}}-1}$ subfiles which corresponds  to a total of $MF$ bits, hence satisfying the memory constraint.
\subsubsection{Delivery phase}
\label{subsubsec:delivery_phase_cen}
During the delivery phase, the tuple $\mathbf{d}$ of all user demands is revealed, where each user $i$ makes a request for file $W_{d_i}$.
Since user $i$ has all subfiles $W_{d_i,\mathcal{T} }$ such that $i \in \mathcal{T}$, the transmitter has to deliver
all subfiles $W_{d_i,\mathcal{T}}$ such that $i \notin \mathcal{T}$, for all users $i \in [K]$.
This corresponds to a total of $K(1-\mu)F$ bits to be delivered over the wireless channel.

The transmitter splits each subfile $W_{d_i,\mathcal{T}}$, with $i \notin \mathcal{T}$,
into a common mini-subfile $W_{d_i,\mathcal{T}}^{(\mathrm{c})}$ and a private mini-subfile $W_{d_i,\mathcal{T}}^{(\mathrm{p})}$ such that
$W_{d_i,\mathcal{T}}=\big(W_{d_i,\mathcal{T}}^{(\mathrm{c})}, W_{d_i,\mathcal{T}}^{(\mathrm{p})}\big)$.
The two mini-subfiles $W_{d_i,\mathcal{T}}^{(\mathrm{c})}$ and $W_{d_i,\mathcal{T}}^{(\mathrm{p})}$ have sizes $q |W_{d_i,\mathcal{T}}|$ bits and $(1-q) |W_{d_i,\mathcal{T}}|$ bits respectively, where $|W_{d_i,\mathcal{T}}|$ is the size of file $W_{d_i,\mathcal{T}}$ and $q$ is the file splitting ratio given by
\begin{equation}
\label{eq:file_split_ratio}
q=\frac{(1+ m_{\mathrm{C}} )(\alpha-\beta)}{K(1-(\alpha-\beta)) + (1+  m_{\mathrm{C}} )(\alpha-\beta)}.
\end{equation}
All common mini-subfiles are coded using the techniques in the original coded-multicasting scheme in \cite{Maddah-Ali2014}.
In particular, subsets of $1+ m_{\mathrm{C}}$ common mini-subfiles $W_{d_i,\mathcal{T}}^{(\mathrm{c})}$ are combined together using a bitwise XOR operation to generate multicasting messages intended for subsets of $1+ m_{\mathrm{C}}$ users as follows
\begin{equation}
\label{eq:XOR}
W^{\mathrm{(c)}}_{\mathcal{S}}=\oplus_{i \in \mathcal{S}}W^{(\mathrm{c})}_{d_i, \mathcal{S} \setminus \{i\}}
\end{equation}
for all $\mathcal{S} \in \Theta$, where $\Theta=\{\mathcal{S} \subseteq [K]: |\mathcal{S}|= 1+ m_{\mathrm{C}}\}$.
All multicasting messages $W^{\mathrm{(c)}}_{\mathcal{S}}$ are encoded into a common codeword  $X^{(\mathrm{c})}$,
while all private mini-subfiles $W_{d_i,\mathcal{T}}^{(\mathrm{p})}$ intended to user $i$ are encoded into the private codeword $X_i^{(\mathrm{p})}$.
Next, the transmission of the common and private codewords over the wireless channel is carried out as described in Appendix \ref{sec:proof_private_common_GDoF}.

By decoding $X^{(\mathrm{c})}$, each user $i$ retrieves the multicasting messages $W^{\mathrm{(c)}}_{\mathcal{S}}$ for all $\mathcal{S} \in \Theta$.
Hence, user $i$ recovers all missing common mini-subfiles by combining with the content of its local cache as in \cite{Maddah-Ali2014}.
For example, for some $\mathcal{T}$ such that $i \notin \mathcal{T}$, user $i$ solves for the missing $W_{d_i,\mathcal{T}}^{(\mathrm{c})}$
using XOR combining of $W^{\mathrm{(c)}}_{\mathcal{S}}$, where $\mathcal{S}=\mathcal{T} \cup \{i\}$, with the pre-stored $m_{\mathrm{C}}$ common mini-subfiles $W_{d_k,\mathcal{S} \setminus \{k\} }^{(\mathrm{c})}$ with  $k \in \mathcal{T}$.
After decoding $X^{(\mathrm{c})}$, and removing its contribution from the received signal as explained in Appendix \ref{sec:proof_private_common_GDoF}, user $i$ decodes the private codeword $X_i^{(\mathrm{p})}$, from which the missing private mini-subfiles $W_{d_i,\mathcal{T}}^{(\mathrm{p})}$, with $\mathcal{T}$  such that $i \notin \mathcal{T}$, are retrieved.
At this stage, the entire requested file $W_{d_i}$ is recovered.
\subsubsection{Achievable Delivery Time}
The shared-link-type transmission, taking place over $X^{(\mathrm{c})}$, delivers a total of $qK(1-\mu)$ files (by excluding the parts already cached)
at rate $(\alpha-\beta)(1+m_{\mathrm{C}})$ files per time slot, where $(\alpha-\beta)$ is the GDoF of the physical channel as seen from Lemma \ref{lemma:private_common_GDoF} and $(1+m_{\mathrm{C}})$ is the gain due to coded-multicasting.
Hence, the delivery time for the shared link layer is
\begin{equation}
\frac{Kq(1-\mu)}{(\alpha-\beta)(1+m_{\mathrm{C}})}=\frac{K \left(1-\mu \right)}{K \left(1-(\alpha-\beta) \right)+ \left(1+ m_{\mathrm{C}} \right)(\alpha-\beta)}.
\end{equation}
On the other hand, each $X_i^{(\mathrm{p})}$ in the zero-forcing layer delivers a total of $(1-q)(1-\mu)$ files at rate
$1-(\alpha-\beta)$ files per time slot, as seen from Lemma \ref{lemma:private_common_GDoF}.
Hence, the delivery time for this layer is
\begin{equation}
\frac{K(1-q)(1-\mu)}{K\big(1-(\alpha-\beta)\big)}=\frac{K \left(1-\mu \right)}{K \left(1-(\alpha-\beta) \right)+ \left(1+ m_{\mathrm{C}} \right)(\alpha-\beta)}.
\end{equation}
Since the two layers take place in parallel, the total delivery time is also given by
\begin{equation}
\mathsf{GNDT}_{\mathrm{C}}(\mu,\alpha,\beta)=\frac{K \left(1-\mu \right)}{K \left(1-(\alpha-\beta) \right)+ \left(1+ m_{\mathrm{C}} \right)(\alpha-\beta)}.
\end{equation}
As $\mathsf{GNDT}_{\mathrm{C}}(\mu,\alpha,\beta)$ is achievable, then the corresponding GDoF given by ${\mathsf{GDoF}}_{\mathrm{C}}(\mu,\alpha,\beta)$ is achievable.
\subsection{Converse} \label{sec:converse_cen}
Here we prove the converse in (\ref{eq:converse_cen}), which is equivalent to showing order-optimality of $\mathsf{GNDT}_{\mathrm{C}}(\mu,\alpha,\beta)$, i.e.
$\mathsf{GNDT}_{\mathrm{C}}(\mu,\alpha,\beta)/\mathsf{GNDT}(\mu,\alpha,\beta) \leq 12$.
Since $\mathsf{GNDT}_{\mathrm{C}}(\mu,\alpha,\beta)$ and $\mathsf{GNDT}^{\mathrm{lb}}_s(\mu,\alpha,\beta)$ only depend on the difference $(\alpha-\beta)$, with a slight abuse of notation we define
\begin{equation} \label{eq:T_cen_delta}
\mathsf{GNDT}_{\mathrm{C}}(\mu,\delta) \triangleq  \frac{K \left(1-\mu \right)}{K \left(1-\delta \right)+ \left(1+K\mu \right)\delta}
\end{equation}
and
\begin{equation} \label{eq:cut_set_bound_s_delta}
\mathsf{GNDT}^{\mathrm{lb}}_s(\mu,\delta) \triangleq \left( \frac{s}{1+(s-1)(1-\delta)}\left(1-\frac{M}{\left \lfloor{\frac{N}{s}} \right  \rfloor}\right) \right)^+.
\end{equation}
where $\delta \in [0,1]$, $\mathsf{GNDT}_{\mathrm{C}}(\mu,\delta = \alpha-\beta) = \mathsf{GNDT}_{\mathrm{C}}(\mu,\alpha,\beta)$
and $\mathsf{GNDT}^{\mathrm{lb}}_s(\mu,\delta=\alpha-\beta) =\mathsf{GNDT}^{\mathrm{lb}}_s(\mu,\alpha,\beta)$.
Note $\mu \in \{0,\frac{1}{K},\frac{2}{K}, \dots, \frac{K-1}{K},1 \}$ is assumed in \eqref{eq:T_cen_delta},
where the lower convex envelope is taken for the remaining points in $\mu \in [0,1]$.
From the above, the lower bound in \eqref{eq:cut_set_bound} is rewritten as
\begin{equation} \label{eq:cut_set_bound_delta}
\mathsf{GNDT}(\mu,\alpha,\beta) \geq \max_{s \in \{1,2,\dots,K\}} \mathsf{GNDT}^{\mathrm{lb}}_s(\mu,\delta=\alpha-\beta)
\end{equation}
In the remaining part, we work with $\mathsf{GNDT}_{\mathrm{C}}(\mu,\delta)$ and  $\mathsf{GNDT}^{\mathrm{lb}}_s(\mu,\delta)$ for convenience.
We show in Appendix \ref{sec:proof_th_cent} that for any $\mu$, there exists a particular $s \in [K]$ such that $\mathsf{GNDT}_{\mathrm{C}}(\mu,\delta)/\mathsf{GNDT}^{\mathrm{lb}}_{s}(\mu,\delta) \leq 12$  for all $\delta \in [0,1]$.
Since the right-hand-side of \eqref{eq:cut_set_bound_delta} is bounded below by
$\mathsf{GNDT}^{\mathrm{lb}}_s(\mu,\delta)$ for any $s \in [K]$, the order-optimality within a factor of $12$ follows.
This concludes the proof of the converse.
\section{Decentralized Placement}
\label{sec:decentralized}
In this section, we prove Theorem \ref{th_decen} which considers the decentralized setting.
As in Section \ref{sec:centralized}, we start with the achievability and then proceed to prove order-optimality.
\subsection{Achievability Scheme}
Here we propose a decentralized scheme
which achieves the delivery time given by $\mathsf{GNDT}_{\mathrm{D}}(\mu,\alpha,\beta)$ in \eqref{eq:T_decen}, and hence the GDoF given by $\mathsf{GDoF}_{\mathrm{D}}(\mu,\alpha,\beta)$ in Theorem \ref{th_decen}.
We start with the placement phase.
\subsubsection{Placement Phase}
This is similar to the procedure in the original decentralized coded-caching paper  \cite{Maddah-Ali2015a},
and hence does not depend on the wireless channel parameters.
Each user $i$ stores a subset of $\mu F$ bits from each file, chosen uniformly at random.
Therefore, each bit of each file is stored in some subset of users\footnote{For a set $\mathcal{S}$, the power set $2^{\mathcal{S}}$ consists of all subsets of $\mathcal{S}$ (including $\mathcal{S}$ itself) and the empty set $\emptyset$. Note that
we consider finite $[K]$, i.e. $K$ does not go to infinity. This guarantees that the power set is not an uncountable set.}
$\tilde{\mathcal{T}} \in 2^{[ K ]}$, where $|\tilde{\mathcal{T}}| \in \{0,1,\ldots,K\}$.
For some $l \in [N]$, we use $W_{l,\tilde{\mathcal{T}}}$ to denote the bits of file $W_l$ which are stored by all
users in ${\tilde{\mathcal{T}}}$, where each $W_{l,\tilde{\mathcal{T}}}$ is referred to as a subfile henceforth.
It is readily seen that $W_l$ can be reconstructed from $\big\{W_{l,\tilde{\mathcal{T}}}: \: \tilde{\mathcal{T}} \in 2^{[ K ]} \big\}$.
\subsubsection{Delivery Phase}
User $i$ requires all subfiles $W_{d_i,\tilde{\mathcal{T}}}$, such that $i \notin \tilde{\mathcal{T}}$, in order to recover the requested file $W_{d_i}$.
The delivery phase takes place over $K$ sub-phases indexed by $m \in \{0,1,\ldots,K-1\}$.
In the $m$-th sub-phase, the transmitter delivers all subfiles $W_{d_i,\tilde{\mathcal{T}}}$, such that
$i \in [K]$ and $i \notin \tilde{\mathcal{T}}$, with $|\tilde{\mathcal{T}}| = m$.
Note that $m$ goes up to $K-1$ since for $|\tilde{\mathcal{T}}| = K$, the corresponding subfiles are pre-stored by all users.

Focusing on the $m$-th delivery sub-phase, delivery is carried out as described in Section \ref{subsubsec:delivery_phase_cen} for the centralized setting, while replacing $m_{\mathrm{C}}$ in Section \ref{subsubsec:delivery_phase_cen} by $m$.
This is due to the fact that each subfile to be delivered during the $m$-th decentralized delivery sub-phase is pre-stored by $m$ users instead of $m_{\mathrm{C}}$ users in centralized delivery.
It follows that coded-multicasting messages have order $1+m$ in the $m$-th decentralized delivery sub-phase compared to $1+m_{\mathrm{C}}$ in centralized delivery, which is due to the random decentralized placement.
Note that when performing the XOR operation in \eqref{eq:XOR} for the decentralized setting, all subfiles are assumed to be zero-padded to the length of the longest subfile \cite{Maddah-Ali2015a}.
By the end of the $K$ delivery sub-phases, the entire requested files  are recovered by the users.

Note that in sub-phase $m = 0$, there are no coded-multicasting opportunities as this sub-phase delivers parts which are not pre-stored by any user.
Hence, the transmission here is similar to the centralized setting with $\mu = 0$, which corresponds to transmission in the classical MISO BC with no caches \cite{Davoodi2018}.
\subsubsection{ Achievable Delivery Time}
Consider the $m$-th sub-phase and an arbitrary subset of users $\tilde{\mathcal{T}}$ with size $m$.
For each file $W_l$, $l \in [N]$, the probability of any of its bits to be stored in the cache of some user in $\tilde{\mathcal{T}}$ is given by $\mu$.
Hence, the probability of this bit to be stored by exactly the $m$ users of  $\tilde{\mathcal{T}}$ is given by $\mu^{m}(1-\mu)^{K-m}$,
from which the expected number of bits stored by each of such users is given by $\mu^{m}(1-\mu)^{K-m}F$.
It follows that, as $F \rightarrow \infty$, the expected size of $W_{l,\tilde{\mathcal{T}}}$ is given by
\begin{equation}
\mu^{m}(1-\mu)^{K-m}F+o(F)
\end{equation}
where the term $o(F)$ is omitted in the following calculations.
Since there is a total of $\binom{K}{m}$ subsets of $m$ users, we have $\binom{K}{m}\mu^{m}(1-\mu)^{K-m}F$ bits of each file which are cached by exactly  $m$ users.

Now we proceed to calculated the number of bits of the file $W_{d_i}$, which are stored by exactly $m$ users, which
have to be delivered to user $i$.
Recall that user $i$ already has all subfiles $W_{d_i,\tilde{\mathcal{T}}}$, with $|\tilde{\mathcal{T}}|=m$ and $i \in \tilde{\mathcal{T}}$,
pre-stored.
Hence, user $i$ already has $\binom{K-1}{m-1}\mu^{m}(1-\mu)^{K-m}F$ bits of $W_{d_i}$ which are cached in exactly $m$ users.
Hence, the number of unavailable bits, contained in all subfiles $W_{d_i,\tilde{\mathcal{T}}}$ with $|\tilde{\mathcal{T}}|=m$ and $i \notin \tilde{\mathcal{T}}$, is given by $\binom{K-1}{m}\mu^{m}(1-\mu)^{K-m}F$.
Since there are $K$ users in total, the total number of files (obtained after normalizing by $F$)
which have to be delivered during the $m$-th sub-phase is given by
\begin{equation}
K \binom{K-1}{m}\mu^{m}(1-\mu)^{K-m}.
\end{equation}
A portion $q(m) = \frac{(1+m)(\alpha-\beta)}{(1+m)(\alpha - \beta) + K(1-(\alpha - \beta))}$ of such files are delivered with coded-multicasting gain   $1+m$ (i.e. simultaneously useful for $1+m$ users) over the common codeword with GDoF $(\alpha - \beta)$ files per time-slot.
On the other hand, the remaining portion of $1-q(m)$ is delivered over the private codewords with GDoF $K \left(1-(\alpha - \beta) \right)$  files per time-slot.
Hence, the delivery time of the $m$-th sub-phase is
\begin{equation}
 \frac{K \binom{K-1}{m} \mu^m\left(1-\mu\right)^{K-m}}{K(1-(\alpha-\beta))+(1+m)(\alpha-\beta)}.
\end{equation}
By summing over all $K$ sub-phases, the total delivery time is given by
\begin{equation}
\mathsf{GNDT}_{\mathrm{D}}(\mu,\alpha,\beta) = K \sum_{m=0}^{K-1}{\frac{\binom{K-1}{m} \mu^m\left(1-\mu\right)^{K-m}}{K(1-(\alpha-\beta))+(1+m)(\alpha-\beta)}}.
\end{equation}
It follows that the corresponding GDoF given by ${\mathsf{GDoF}}_{\mathrm{D}}(\mu,\alpha,\beta)$ is achievable.
\subsection{Converse} \label{sec:proof_th_opt_decen}
In this part, we prove the converse in (\ref{eq:converse_decen}), which is equivalent to showing order-optimality of $\mathsf{GNDT}_{\mathrm{D}}(\mu,\alpha,\beta)$, i.e.
 $\mathsf{GNDT}_{\mathrm{D}}(\mu,\alpha,\beta)/\mathsf{GNDT}(\mu,\alpha,\beta) \leq 12$.
As in the centralized setting, $\mathsf{GNDT}_{\mathrm{D}}(\mu,\alpha,\beta)$ only depends on the difference $\delta=(\alpha-\beta)$. Therefore, we work with
\begin{equation} \label{eq:T_decen_delta}
\mathsf{GNDT}_{\mathrm{D}}(\mu,\delta ) \triangleq K \sum_{m=0}^{K-1}{\frac{\binom{K-1}{m} \mu^m\left(1-\mu\right)^{K-m}}{K(1-\delta)+(1+m)\delta}}
\end{equation}
where $\mathsf{GNDT}_{\mathrm{D}}(\mu,\delta = \alpha-\beta ) = \mathsf{GNDT}_{\mathrm{D}}(\mu,\alpha,\beta)$.
Unlike $\mathsf{GNDT}_{\mathrm{C}}(\mu,\delta )$  in \eqref{eq:T_cen_delta}, $\mathsf{GNDT}_{\mathrm{D}}(\mu,\delta )$
does not have the desirable form which allows comparing it to the bound in  \eqref{eq:cut_set_bound_delta} directly.
Hence, the first (key) step of the converse is to derive an upper bound on $\mathsf{GNDT}_{\mathrm{D}}(\mu,\delta)$, denoted by $\mathsf{GNDT}^{\mathrm{ub}}_{\mathrm{D}}(\mu,\delta)$,
which takes the form of the centralized achievable delivery time in \eqref{eq:T_cen_delta}.
This is given in the following result.
\begin{lemma} \label{lemma_ineq_decen}
The decentralized delivery time  $\mathsf{GNDT}_{\mathrm{D}}(\mu,\delta)$ is bounded above as
	\begin{equation} \label{eq:decent_inequality_1}
	\mathsf{GNDT}_{\mathrm{D}}(\mu,\delta)  \leq \mathsf{GNDT}^{\mathrm{ub}}_{\mathrm{D}}(\mu,\delta)
= \frac{K \left(1-\mu \right)}{K(1-\delta) + (1+u)\delta}
	\end{equation}
	where $u$ is given by
\begin{equation} \label{eq:eq_intr_u}
u =\frac{K \left(1-\mu \right)}{\mathsf{GNDT}_{\mathrm{D}}(\mu,1)} - 1.
\end{equation}
\end{lemma}
The proof of Lemma \ref{lemma_ineq_decen} is given in Appendix \ref{sec:proof_lemma_ineq_decen}.
One important consequence of Lemma \ref{lemma_ineq_decen} is that the expression in \eqref{eq:decent_inequality_1}
allows us to show order-optimality of  $\mathsf{GNDT}_{\mathrm{D}}^{\mathrm{ub}}(\mu,\delta)$ to within a factor of $12$
using similar techniques to the ones used for the centralized setting.
The details are relegated to Appendix \ref{sec:proof_order_opt_T}.
The order-optimality of $\mathsf{GNDT}_{\mathrm{D}}(\mu,\delta)$ to within a factor of $12$ follows, which concludes the converse.
\subsection{Gap Between Decentralized and Centralized Schemes} \label{subsec:gap_cen_decen}
From a straightforward inspection of \eqref{eq:T_cen_delta} and \eqref{eq:decent_inequality_1}, it can be seen that
for integer values of $K\mu$ (for which a close form of $\mathsf{GNDT}_{\mathrm{C}}(\mu,\delta)$ is obtained), we have
\begin{equation} \label{eq:eq_comparison_cen_decen}
\frac{\mathsf{GDoF}_{\mathrm{C}}(\mu,\delta)}{\mathsf{GDoF}_{\mathrm{D}}(\mu,\delta)} = \frac{\mathsf{GNDT}_{\mathrm{D}}(\mu,\delta)}{\mathsf{GNDT}_{\mathrm{C}}(\mu,\delta)} \leq
\frac{K(1-\delta)+(K\mu+1)\delta}{K(1-\delta)+(u+1)\delta}.
\end{equation}
We know that when $\delta=1$ (i.e. $\alpha=1$ and $\beta=0$),
all spatial multiplexing gains are lost and the achievable delivery times collapse to the ones in \cite{Maddah-Ali2014,Maddah-Ali2015a}.
Hence, it follows from the observations in \cite{Maddah-Ali2015a} (and then the proof in \cite{Yan2016}) that for $\delta=1$,
there is a small price to pay due to decentralization, making the ratio in \eqref{eq:eq_comparison_cen_decen} small.
By further examining the bound on the right-most side of \eqref{eq:eq_comparison_cen_decen},
it can be seen that it decreases when $\delta$ decreases, hence further reducing the price of decentralization.
For example, such price is minimal when $\delta=0$ (i.e. $\alpha=\beta$), where both the centralized and decentralized strategies achieve
the optimal delivery $\mathsf{GNDT}(\mu, \alpha, \beta=\alpha)=1-\mu$.
This is intuitive as with a decreased $\delta$, the system starts to rely more on spatial multiplexing gains and local caching gains and less on global caching gains, which are affected by decentralization.
Concretely, the gap in \eqref{eq:eq_comparison_cen_decen} is bounded above as follows.
\begin{corollary} \label{corollary_cen_decen}
For any $\delta \in [0,1]$ and $\mu \in [0,1]$, we have
	\begin{equation} \label{eq:cen_decen_corollary}
	\frac{\mathsf{GDoF}_{\mathrm{C}}(\mu,\delta)}{\mathsf{GDoF}_{\mathrm{D}}(\mu,\delta)} \leq 1.5.
	\end{equation}
\end{corollary}
The above corollary is obtained by employing the results in
Theorem \ref{th_cen}, Lemma \ref{lemma_ineq_decen} and  \cite{Yan2016}.
The full proof is relegated to Appendix \ref{appendix:proof_corollary}.
\section{Conclusions}
\label{sec:conclusion}
In this paper, we characterized the optimal GDoF of the symmetric cache-aided MISO BC under partial CSIT up to a constant multiplicative factor.
Moreover, we showed that such GDoF characterization is robust to decentralization, i.e. we proposed a decentralized caching
strategy which attains an order-optimal GDoF performance.
In order to derive the GDoF results, we introduced the generalized normalized delivery time (GNDT) metric, which extends the normalized delivery time (NDT) metric in the same way the GDoF extends the DoF.
The GNDT is related to the reciprocal of the GDoF, and is generally easier to deal with when characterizing achievable and optimal performances.

At the heart of our converse proof is a GDoF outer bound for a parallel MISO BC with partial CSIT,
which extends a family of robust outer bounds based on the aligned image sets approach, initially developed in the context of classical networks with no caches, to cache-aided networks.
On the other hand, we showed that the order optimal GDoF takes a familiar weighted-sum form, often observed in classical networks (with no caches) under partial CSIT.
Achieving such GDoF relies on a key interplay between spatial multiplexing and coded-multicasting gains.

This work opens the door for  a number  of interesting extensions.
An intriguing direction is to consider a setting in which each transmitter can only store part of the library, hence enabling only partial transmitter cooperation as opposed to the full cooperation assumed in this work (e.g. general cache-aided interference network).
This setting generalizes the works in \cite{Xu2017,Hachem2018} to the GDoF framework under partial CSIT considered in this paper.
Progress along these lines is reported in \cite{Lampiris2017}, while limiting to absent CSIT and considering only achievability.
As observed in \cite{Xu2017,Hachem2018}, under such partial cooperating (through caching at the transmitters), the underlying physical channel is modeled by the X channel.
Hence, it is worthwhile highlighting in this context that for the X channel, the GDoF under partial CSIT  is still an open problem.
Another interesting direction is relaxing the symmetry in the channel.
However, one major difficulty here is the potential explosion in the number of channel parameters.
Therefore it is not surprising that such asymmetric GDoF characterizations are still open even in classical networks \cite{Davoodi2017,Davoodi2018}.
Last but no least, reducing the constant multiplicative factor of 12 is also of significant interest.
For the original shared-link setting, recent efforts managed to reduce the constant multiplicative factor \cite{Ghasemi2017,Yu2017}.
Our observations through numerical simulations, which show that the gap is much smaller than 12, provide hope that such
tightening may also be possible for the order-optimal characterizations presented in this paper.
\appendices
\section{Proof of Lemma \ref{lemma:AIS}}
\label{sec:proof_AIS_lemma}
The proof is based on the approach in \cite{Davoodi2016,Davoodi2017a,Davoodi2018}, where outer bounds under finite precision and partial CSIT are derived.
We follow the same overall steps in these works, while specializing to the specific setup considered here.
For simplicity and notational briefness, we focus on real channels.
The extension to complex channels follows along the lines of \cite{Davoodi2016,Davoodi2017a}.
We consider $s = K$ users. For general $s \leq K$, the exact same steps follow while considering only the corresponding $s$ rate bounds.
\subsection{Deterministic Channel Model}
The first step is to convert the channel into a deterministic equivalent with inputs and outputs all being integers.
This is given by
\begin{align}
\label{eq:deterministic_model_parallel_1}
\bar{Y}_{i}(t) &=  \lfloor G_{ii}(t) \bar{X}_{i}(t) \rfloor +  \sum_{j \in [K] \setminus \{i\} }
\lfloor  \bar{P}^{\alpha - 1} G_{ij}(t)  \bar{X}_{j}(t)\rfloor \\
\label{eq:deterministic_model_parallel_2}
\bar{B}_{i}(t) &= \bar{A}_{i}(t)
\end{align}
where $\bar{P} = \sqrt{P}$,
$\bar{X}_{i}(t)\in \{0,1,\ldots, \lfloor \bar{P} \rfloor  \}$ and
$\bar{A}_{i}(t) \in \{0,1,\ldots, \lfloor \bar{P}^{\gamma} \rfloor  \}$,
$\forall i \in [K]$.
It can be shown that a GDoF upper bound for the deterministic channel is also a GDoF upperbound
for the original channel using the same steps in \cite{Davoodi2016}.
Therefore we focus on the deterministic channel henceforth.
\subsection{Fano’s Inequality and Differences of Entropies}
For notational brevity, we define $M_{i} \triangleq \big( W_{d_{i}^{1}},\ldots,W_{d_{i}^{\left \lfloor{{N}/{s}}\right \rfloor}} \big)$ to denote the set of messages to be delivered to user $i$.
Moreover, we define $M_{[i:K]} \triangleq M_{i},\ldots,M_{K}$.
Using Fano's inequality, for user $k$ we have
\begin{align}
nR_{k} & \leq I \left(M_{k} ; \bar{Y}_{k}^{n},\bar{B}_{k}^{n} \mid M_{[k+1:K]}, \mathcal{G} \right) + o(n) \\
& \leq H \left(\bar{Y}_{k}^{n},\bar{B}_{k}^{n} \mid M_{[k+1:K]}, \mathcal{G} \right)
- H \left(\bar{Y}_{k}^{n},\bar{B}_{k}^{n} \mid M_{[k:K]}, \mathcal{G} \right) + o(n).
\end{align}
After omitting $o(n)$ and $o\left( \log(P) \right)$  terms, we obtain
\begin{equation}
\label{eq:parallel_sum_rate_upperbound}
n\sum_{k = 1}^{K} R_{k} \leq n (1 + \gamma) \log(\bar{P}) +
\sum_{k=2}^{K} \underbrace{H \left(\bar{Y}_{k-1}^{n},\bar{B}_{k-1}^{n} \mid M_{[k:K]}, \mathcal{G} \right)
	- H \left(\bar{Y}_{k}^{n},\bar{B}_{k}^{n} \mid M_{[k:K]}, \mathcal{G} \right)}_{H^{\Delta}_{k}}.
\end{equation}
Hence, the focus becomes to bound the differences of entropies $H^{\Delta}_{2},\ldots,H^{\Delta}_{K}$.
\subsection{Bounding the Differences of Entropies}
Focusing on the term $H^{\Delta}_{k}$, $k \in [2:K]$, we proceed  as follows:
\begin{align}
\label{eq:H_delta_ineq_1}
H^{\Delta}_{k}  = \ & H \left(\bar{Y}_{k-1}^{n},\bar{B}_{k-1}^{n} \mid M_{[k:K]}, \mathcal{G} \right)
	- H \left(\bar{Y}_{k}^{n},\bar{B}_{k}^{n} \mid M_{[k:K]}, \mathcal{G} \right) \\
\nonumber
 =  \ &  H \left(\bar{Y}_{k-1}^{n} \mid M_{[k:K]}, \mathcal{G} \right)
	- H \left(\bar{Y}_{k}^{n} \mid M_{[k:K]}, \mathcal{G} \right) \\
\label{eq:H_delta_ineq_2}
   &+
H \left(\bar{B}_{k-1}^{n} \mid M_{[k:K]}, \mathcal{G},\bar{Y}_{k-1}^{[n]} \right)
	- H \left(\bar{B}_{k}^{n} \mid M_{[k:K]}, \mathcal{G}, \bar{Y}_{k}^{n}\right) \\
\label{eq:H_delta_ineq_3}
\leq \ &  H \left(\bar{Y}_{k-1}^{n} \mid M_{[k:K]}, \mathcal{G} \right)
	- H \left(\bar{Y}_{k}^{n} \mid M_{[k:K]}, \mathcal{G} \right) + n \log \big(\bar{P}^{\gamma} + 1\big).
\end{align}
In the above, \eqref{eq:H_delta_ineq_2} is obtained from the chain rule, while
\eqref{eq:H_delta_ineq_3} follows from $H \big(\bar{B}_{k}^{n} \mid M_{[k:K]}, \mathcal{G}, \bar{Y}_{k}^{n}\big)  \geq 0$
and $H \big(\bar{B}_{k-1}^{n} \mid M_{[k:K]}, \mathcal{G},\bar{Y}_{k-1}^{n} \big) \leq
H \big(\bar{B}_{k-1}^{n} \big) \leq \sum_{t = 1}^{n} H \big(\bar{B}_{k-1}(t) \big)\leq n \log \big(\bar{P}^{\gamma} + 1\big)$.
Now it remains to bound the difference of entropies
$H \left(\bar{Y}_{k-1}^{n} \mid M_{[k:K]}, \mathcal{G} \right)
	- H \left(\bar{Y}_{k}^{n} \mid M_{[k:K]}, \mathcal{G} \right)$
under partial CSIT and the bounded density assumptions as described in Section \ref{subsubsec:partial_CSIT}.
This difference is bounded above as
\begin{equation}
\label{eq:H_delta_ineq_4}
H \left(\bar{Y}_{k-1}^{n} \mid M_{[k:K]}, \mathcal{G} \right)
	- H \left(\bar{Y}_{k}^{n} \mid M_{[k:K]}, \mathcal{G} \right) \leq
n\big(1 - (\alpha - \beta) \big)\log(\bar{P}) + o \big(  \log(\bar{P}) \big).
\end{equation}
The inequality in \eqref{eq:H_delta_ineq_4} follows directly from \cite{Davoodi2018} (see the proofs of \cite[Th. 1]{Davoodi2018} and
\cite[Th. 2]{Davoodi2018}), and is obtained using the aligned image sets approach \cite{Davoodi2016}.
Intuitively, under perfect CSIT (i.e. $\beta = \alpha$), the transmitter uses zero-forcing to create a maximal difference of entropies, in a GDoF sense, between $\bar{Y}_{k-1}^{n}$ and $\bar{Y}_{k}^{n}$.
On the other hand, when all paths have equal strengths and the CSIT is limited to finite precision (i.e. $\alpha = 1$ and $\beta = 0$),
a positive difference of entropies in a GDoF sense cannot be created.
Between the two extremes, the transmitter benefits from path-loss and partial CSIT, through power control and zero-forcing,
to create a positive difference of entropies which is bounded above by  $1$, in a GDoF sense.

By combining the bounds in \eqref{eq:H_delta_ineq_4} and  \eqref{eq:H_delta_ineq_3}, we obtain
\begin{equation}
\label{eq:H_delta_ineq_5}
H^{\Delta}_{k} \leq n\big(\gamma + 1 - (\alpha - \beta) \big)\log(\bar{P}) + o \big(  \log(\bar{P}) \big).
\end{equation}
The bound in \eqref{eq:H_delta_ineq_5} holds for all $k \in [2:K]$. By plugging \eqref{eq:H_delta_ineq_5} into \eqref{eq:parallel_sum_rate_upperbound}, the result in \eqref{eq:GDoF_parallel} directly follows.
\section{Proof of Lemma \ref{lemma:private_common_GDoF}}
\label{sec:proof_private_common_GDoF}
First, let us rewrite the signal model in \eqref{eq:signal_model_2} in vector form as
\begin{equation}
Y_{i} = \sqrt{P} \:[\hat{G}_{i1} \cdots \hat{G}_{iK}] \: \mathbf{Q}_i \: \mathbf{X} + \sqrt{P^{1-\beta}} \: [\tilde{G}_{i1} \cdots \tilde{G}_{iK}] \: \mathbf{Q}_i \: \mathbf{X} + Z_i
\end{equation}
where $\mathbf{X} \triangleq [ X_{1}  \cdots  X_{K}]^{\Trn}$ is the signal transmitted from the $K$ transmitters
and $\mathbf{Q}_i$ is a $K \times K$ diagonal matrix
with $1$ as the $(i,i)$-th entry and $\sqrt{P^{\alpha-1}}$ as  the remaining diagonal entries.
Note that we ignore the time index for brevity.
The messages $W^{(\mathrm{c})}$ and $W^{(\mathrm{p})}_1, \dots, W^{(\mathrm{p})}_K$ are encoded into
unit power independent Gaussian codewords
$X^{(\mathrm{c})}$ and $X^{(\mathrm{p})}_1, \dots, X^{(\mathrm{p})}_K$, respectively.
The transmitted signal is then constructed as
\begin{equation}
\mathbf{X}=
\mathbf{D}\left( \sqrt{1-P^{\beta-\alpha}} \mathbf{V}^{(\mathrm{c})}X^{(\mathrm{c})} +  \sqrt{P^{\beta-\alpha}} \sum_{k=1}^{K}{\mathbf{V}_k^{(\mathrm{p})}}X_k^{(\mathrm{p})} \right).
\end{equation}
In the above, $\mathbf{D}$ is a $K \times K$ diagonal matrix where the
$(j,j)$-th entry is $O(1)$ in $P$, and is chosen such that the
power constraint  $\E \left(|X_{j}|^{2} \right) \leq 1$ is not violated.
$\mathbf{V}^{(\mathrm{c})}$ is a generic (random) unit vector and
$\mathbf{V}_k^{(\mathrm{p})}  \triangleq \Big[ V^{(\mathrm{p})}_{k1} \;
\cdots  V^{(\mathrm{p})}_{kK} \; \Big]^{\Trn}$
is a zero-forcing unit vector designed using the channel estimates such that
\begin{equation}
\sqrt{P^{\alpha}} \left( \hat{G}_{i1} V^{(\mathrm{p})}_{k1} + \cdots + \sqrt{P^{1-\alpha}}\hat{G}_{ii} V^{(\mathrm{p})}_{ki} + \cdots + \hat{G}_{iK} V^{(\mathrm{p})}_{kK} \right)=0, \ \forall i \neq k.
\end{equation}
It is simple to verify from the zero-forcing condition that $V^{(\mathrm{p})}_{ki}$ cannot scale faster than
$O(\sqrt{P^{\alpha-1}})$ for all $k \neq i$.
Hence, the received signal of user $i$ is rewritten as
\begin{equation}
Y_{i} = \sqrt{P} a_{i}^{(\mathrm{c})} X^{(\mathrm{c})} + \sqrt{P^{1+\beta-\alpha}} a_{ii}^{(\mathrm{p})} X_{i}^{(\mathrm{p})}  +  \sum_{k=1, k \neq i}^K a_{ik}^{(\mathrm{p})} X_k^{(\mathrm{p})} + Z_i
\end{equation}
where $a_{i}^{(\mathrm{c})}$ and  $a_{ik}^{(\mathrm{p})}$, for all $i,k \in [K]$, are all $O(1)$.

Each user $i$ decodes $X^{(\mathrm{c})}$ by treating interference as noise and recovers $W^{(\mathrm{c})}$.
As $X^{(\mathrm{c})}$ is received with power $O(P)$, while interference plus noise has  power $O(P^{1+\beta-\alpha})$, it follows  that $X^{(\mathrm{c})}$ supports a rate of $ (\alpha-\beta)\log (P) + o\big( \log (P) \big)$.
Then, each user $i$ proceeds to remove the contribution of $X^{(\mathrm{c})}$ from the received signal
and decodes its own  $X_i^{(\mathrm{p})}$ while treating the remaining interference as noise, from which $W_i^{(\mathrm{p})}$ is recovered.
As $X_i^{(\mathrm{p})}$ is received with power $O(P^{1+\beta-\alpha})$, while the remaining interference plus noise has power $O(1)$,
it follows that $X_i^{(\mathrm{p})}$  supports a rate of $(1+\beta-\alpha) \log (P) + o\big( \log (P) \big)$.
\begin{remark}
	It is worthwhile highlighting that the achievable GDoF in Lemma \ref{lemma:private_common_GDoF} (shown in this appendix) can be inferred from \cite{Davoodi2018}.
	One key difference, however, is that the MISO BC considered in \cite{Davoodi2018} has private messages only, and rate-splitting is used to multicast part of the private messages as a common codeword decoded by all users. This relationship between the MISO BC with private messages and its counterpart with a common message under partial CSIT was first observed in \cite{Yang2013}.
\end{remark}
\section{Proofs of Order Optimality}
Here we provide proofs for the order-optimality parts of Theorem \ref{th_cen} and Theorem \ref{th_decen}.
We start with an instrumental lemma used throughout the proofs in the following subsections.
\begin{lemma} \label{lemma_f}
For parameters $K,\mu$ and $s$ defined previously, if $\frac{K}{s (1+K\mu)} \geq 1$, then the function given by
\begin{equation}
f(\delta;K,\mu,s)=\frac{1+(s-1)(1-\delta)}{K(1-\delta)+(1+K\mu)\delta}
\end{equation}
is non-decreasing in $\delta \in [0,1]$.
\end{lemma}
\begin{proof}
The derivative of $f(\delta;K,\mu,s)$ with respect to $\delta$ is given by $\frac{df}{d\delta}=-\frac{s(1+K\mu)-K}{(K(1-\delta)+(1+K\mu)\delta)^2}$,
which is non-negative for $K\geq s(1+K\mu)$.
\end{proof}
\subsection{Order Optimality of $\mathsf{GNDT}_{\mathrm{C}}(\mu,\delta)$}
\label{sec:proof_th_cent}
We show here that for any $\mu$, there exists a particular $s \in [K]$ such that  $\mathsf{GNDT}_{\mathrm{C}}(\mu,\delta)/\mathsf{GNDT}^{\mathrm{lb}}_{s}(\mu,\delta) \leq 12$  for all $\delta \in [0,1]$.
We handle the two cases $K \leq 12$ and $K \geq 13$ separately.
Starting with $K \leq 12$, consider a generic \mbox{$\delta \in [0,1]$}.
By setting $s=1$ in (\ref{eq:cut_set_bound_s}), we get that $\mathsf{GNDT}_1^{\mathrm{lb}}(\mu,\delta) = 1-\mu$.
On the other hand, \mbox{$\mathsf{GNDT}_{\mathrm{C}}(\mu,\delta) \leq \mathsf{GNDT}_{\mathrm{C}}(\mu,1) \leq K(1-\mu) $}.
Hence,
 ${\mathsf{GNDT}_{\mathrm{C}}(\mu,\delta)}/{\mathsf{GNDT}^{\mathrm{lb}}_{1}(\mu,\delta)} \leq 12$.

Next, we consider $K \geq 13$. As in \cite{Maddah-Ali2014}, we split the problem in three sub-cases: the sub-case \mbox{$0 \leq \mu \leq \frac{1.1}{K}$}, the sub-case \mbox{$\frac{1.1}{K} < \mu  \leq 0.092$} and the sub-case $0.092 < \mu \leq 1$.
We start with $0 \leq \mu \leq \frac{1.1}{K}$.
For $\delta=1$, we have  $\mathsf{GNDT}_{\mathrm{C}}(\mu,1) \leq \mathsf{GNDT}_{\mathrm{C}}(0,1) = K$.
By setting $s = \left \lfloor {0.275K} \right \rfloor$, we know from \cite{Maddah-Ali2014} that
${\mathsf{GNDT}^{\mathrm{lb}}_{s}(\mu,1)} \geq K/12.$
On the other hand, for a generic $\delta \in [0,1]$, the following upper bound holds
\begin{equation} \frac{\mathsf{GNDT}_{\mathrm{C}}(\mu,\delta)}{\mathsf{GNDT}^{\mathrm{lb}}_{s}(\mu,\delta)} \leq \frac{\mathsf{GNDT}_{\mathrm{C}}(0,\delta)}{\mathsf{GNDT}^{\mathrm{lb}}_{s}(\mu,\delta)} = \underbrace{\frac{1+(s-1)(1-\delta)}{K(1-\delta)+\delta}}_{f(\delta;K,0,s)} \cdot \frac{K}{s\left(1-\frac{M}{\left \lfloor{\frac{N}{s}} \right  \rfloor}\right)}.
\end{equation}
Since $\frac{K}{s} \geq \frac{1}{0.275}> 1$, from Lemma \ref{lemma_f} it follows that $f(\delta;K,0,s)$ is non-decreasing in $\delta \in [0,1]$.
Hence,
\begin{equation}
\frac{\mathsf{GNDT}_{\mathrm{C}}(\mu,\delta)}{\mathsf{GNDT}^{\mathrm{lb}}_{s}(\mu,\delta)} \leq \frac{\mathsf{GNDT}_{\mathrm{C}}(0,\delta)}{\mathsf{GNDT}^{\mathrm{lb}}_{s}(\mu,\delta)} \leq \frac{\mathsf{GNDT}_{\mathrm{C}}(0,1)}{\mathsf{GNDT}^{\mathrm{lb}}_{s}(\mu,1)} \leq 12.
\end{equation}
We proceed to the sub-case $\frac{1.1}{K} < \mu  \leq 0.092$.
Let $\tilde{\mu}$ be the largest number in $[0,\mu]$ such that $K\tilde{\mu}$ is an integer.
We know from \cite{Maddah-Ali2014} that
\mbox{$\mathsf{GNDT}_{\mathrm{C}}(\mu,1) \leq  \mathsf{GNDT}_{\mathrm{C}}(\tilde{\mu},1) \leq \frac{1}{\mu}.$}
By setting $s= \left \lfloor {\frac{0.3}{\mu}} \right \rfloor$, we also know from \cite{Maddah-Ali2014} that
$ {\mathsf{GNDT}^{\mathrm{lb}}_{s}(\mu,1)} \geq \frac{1}{12 \mu}. $
Considering a generic $\delta \in [0,1]$, we write
\begin{equation}
\frac{\mathsf{GNDT}_{\mathrm{C}}( \mu,\delta)}{\mathsf{GNDT}^{\mathrm{lb}}_{s}( \mu,\delta)} \leq \frac{\mathsf{GNDT}_{\mathrm{C}}(\tilde{\mu},\delta)}{\mathsf{GNDT}^{\mathrm{lb}}_{s}( \mu,\delta)} = \underbrace{\frac{1+(s-1)(1-\delta)}{K(1-\delta)+(1+K\tilde{\mu})\delta}}_{f(\delta;K,\tilde{\mu},s)} \cdot \frac{K \left(1-\tilde{\mu}\right)}{s\left(1-\frac{M}{\left \lfloor{\frac{N}{s}} \right  \rfloor}\right)}.
\end{equation}
As $\frac{K}{s(1+K\tilde{\mu})} \geq \frac{1}{0.3} \frac{K \mu}{1+K \mu} > 1$, Lemma \ref{lemma_f}  implies that $f(\delta;K,\tilde{\mu},s)$ is non-decreasing in $\delta \in [0,1]$.
Hence,
\begin{equation} \frac{\mathsf{GNDT}_{\mathrm{C}}(\mu,\delta)}{\mathsf{GNDT}^{\mathrm{lb}}_{s}( \mu,\delta)} \leq \frac{\mathsf{GNDT}_{\mathrm{C}}(\tilde{\mu},\delta)}{\mathsf{GNDT}^{\mathrm{lb}}_{s}(\mu,\delta)} \leq \frac{\mathsf{GNDT}_{\mathrm{C}}(\tilde{\mu},1)}{\mathsf{GNDT}^{\mathrm{lb}}_{s}(\mu,1)} \leq 12.
\end{equation}

Finally, we look at the sub-case $0.092 < \mu \leq 1$ and we consider a generic $\delta \in [0,1]$.
By setting $s=1$, we get $ \mathsf{GNDT}_1^{\mathrm{lb}}(\mu,\delta) = 1-\mu$.
Moreover, from \cite{Maddah-Ali2014}, we know that
$ {\mathsf{GNDT}_{\mathrm{C}}(\mu,\delta)} \leq \mathsf{GNDT}_{\mathrm{C}}(\mu,1) \leq 12(1-\mu)$.
Hence
${\mathsf{GNDT}_{\mathrm{C}}(\mu,\delta)}/{\mathsf{GNDT}^{\mathrm{lb}}_{1}(\mu,\delta)} \leq 12$.
This concludes the proof.
\subsection{Order Optimality of $\mathsf{GNDT}_{\mathrm{D}}(\mu,\delta)$}
\label{sec:proof_order_opt_T}
As for the centralized setting, we show that for any $\mu$, there exists a particular $s \in [K]$ such that $\mathsf{GNDT}^{\mathrm{ub}}_{\mathrm{D}}(\mu,\delta)/\mathsf{GNDT}^{\mathrm{lb}}_{s}(\mu,\delta) \leq 12$  for all $\delta \in [0,1]$.
We start with the following lemma.
\begin{lemma} \label{lemma_bound_u}
The value $u$, defined in \eqref{eq:eq_intr_u}, satisfies  $u \leq K \mu$ for all $\mu \in [0,1)$.
\end{lemma}
\begin{proof}
    We focus on $\mu > 0$ as  $u=0$ for $\mu=0$.
    By definition of $u$ in \eqref{eq:eq_intr_u}, we have
    \begin{equation}
    \frac{K(1-\mu)}{K\mu}\left( 1-\left(1-\mu\right)^K  \right) = \frac{K(1-\mu)}{1+u}
    \end{equation}
    which follows from $\mathsf{GNDT}_{\mathrm{D}}(\mu,1)=\frac{K(1-\mu)}{K\mu}\left( 1-\left(1-\mu\right)^K  \right)$, as shown in \cite{Maddah-Ali2015a}.
	Hence, showing that $u \leq K\mu$ it is equivalent to showing that
	\begin{align}
	\frac{K(1-\mu)}{K\mu}\left(1-(1-\mu)^K \right) & \geq \frac{K(1-\mu)}{1+K\mu} \\
\Rightarrow (K\mu+1)\left( 1-(1-\mu)^K \right)& \geq K\mu \\
\label{lemma_bound_u_inequality}
\Rightarrow 1  \geq (K\mu+1)(1-\mu)^K &
	\end{align}
     The inequality in \eqref{lemma_bound_u_inequality} is shown to hold by observing that
    $\mu > 0$ and
	$K\mu+1 \leq (1+\mu)^K$, from which we obtain
	$(K\mu+1)(1-\mu)^K \leq (1+\mu)^K (1-\mu)^K = (1-\mu^2)^K \leq 1.$
	Hence, $u \leq K\mu$ holds.
\end{proof}
Equipped with Lemma \ref{lemma_bound_u}, the remainder of the proof follows the same procedures  in Appendix \ref{sec:proof_th_cent}.
In particular, we consider the two cases $K \leq 12$ and $K \geq 13$.
For the case $K \leq 12$, by setting $s=1$ in (\ref{eq:cut_set_bound_s}), we get that $\mathsf{GNDT}_1^{\mathrm{lb}}(\mu,\delta) = 1-\mu$.
On the other hand, we have $\mathsf{GNDT}^{\mathrm{ub}}_{\mathrm{D}}(\mu,\delta) \leq \mathsf{GNDT}^{\mathrm{ub}}_{\mathrm{D}}(\mu,1) \leq K(1-\mu) $.
It follows that ${\mathsf{GNDT}^{\mathrm{ub}}_{\mathrm{D}}(\mu,\delta)}/{\mathsf{GNDT}^{\mathrm{lb}}_{1}(\mu,\delta)} \leq 12$.

Next, we focus on $K \geq 13$. As in \cite{Maddah-Ali2015a}, we consider three separate sub-cases: the sub-case \mbox{$0 \leq \mu \leq 1/K$}, the sub-case \mbox{$1/K < \mu \leq  1/12$} and the sub-case \mbox{$1/12 < \mu \leq 1$}.
We look at the sub-case \mbox{$0 \leq \mu \leq 1/K$} first.
For $\delta=1$, we have $\mathsf{GNDT}^{\mathrm{ub}}_{\mathrm{D}}(\mu,1) \leq K$,
and by setting $s= \left \lfloor {K/4} \right \rfloor$, we obtain
${\mathsf{GNDT}^{\mathrm{lb}}_{s}(\mu,1)}\geq \frac{1}{12}K$ from \cite{Maddah-Ali2015a}.
On the other hand, for  a generic $\delta \in [0,1]$, we have
\begin{equation} \frac{\mathsf{GNDT}^{\mathrm{ub}}_{\mathrm{D}}(\mu,\delta)}{\mathsf{GNDT}^{\mathrm{lb}}_{s}(\mu,\delta)} =\underbrace{\frac{1+(s-1)(1-\delta)}{K(1-\delta)+(1+u)\delta}}_{f(\delta;K,u/K,s)} \cdot \frac{K \left(1- \mu\right)}{s\left(1-\frac{M}{\left \lfloor{\frac{N}{s}} \right  \rfloor}\right)}.
\end{equation}
By applying Lemma \ref{lemma_bound_u} to lower bound the value of $u$, we can write
$\frac{K}{s(1+u)} \geq \frac{K}{\frac{K}{4} \cdot (1+ K \mu)} > 1$.
Hence, from Lemma \ref{lemma_f}, the function $f(\delta;K,u/K,s)$ is non-decreasing in $\delta \in [0,1]$.
It follows that
\begin{equation}
\label{eq:bounding_decen_inequality}
\frac{\mathsf{GNDT}^{\mathrm{ub}}_{\mathrm{D}}(\mu,\delta)}{\mathsf{GNDT}^{\mathrm{lb}}_{s}(\mu,\delta)} \leq \frac{\mathsf{GNDT}^{\mathrm{ub}}_{\mathrm{D}}(\mu,1)}{\mathsf{GNDT}^{\mathrm{lb}}_{s}(\mu,1)} \leq 12.
\end{equation}

Next, we consider the sub-case $\frac{1}{K} < \mu  \leq \frac{1}{12}$.
From \cite{Maddah-Ali2015a}, we have
$ \mathsf{GNDT}^{\mathrm{ub}}_{\mathrm{D}} (\mu,1) \leq \frac{1}{\mu} $,
and by setting $s= \left \lfloor {\frac{1}{4\mu}} \right \rfloor$, we have
$ {\mathsf{GNDT}^{\mathrm{lb}}_{s}(\mu,1)} \geq \frac{1}{12\mu}$.
For a generic $\delta \in [0,1]$, we have
\begin{equation}
\frac{\mathsf{GNDT}^{\mathrm{ub}}_{\mathrm{D}}(\mu,\delta)}{\mathsf{GNDT}^{\mathrm{lb}}_{s}(\mu,\delta)} =\underbrace{\frac{1+(s-1)(1-\delta)}{K(1-\delta)+(1+u)\delta}}_{f(\delta;K,u/K,s)} \cdot \frac{K \left(1- \mu\right)}{s\left(1-\frac{M}{\left \lfloor{\frac{N}{s}} \right  \rfloor}\right)}.
\end{equation}
By applying Lemma \ref{lemma_bound_u}, it follows that
$\frac{K}{s(1+u)} \geq 4 \cdot \frac{K\mu}{1+K\mu} > 1$.
Hence, from Lemma \ref{lemma_f}, $f(\delta;K,u/K,s)$ is non-decreasing in $\delta \in [0,1]$.
Therefore, the statement in \eqref{eq:bounding_decen_inequality} holds here as well.

Finally, we consider the remaining sub-case $1/12 < \mu \leq 1$ for a generic $\delta \in [0,1]$.
By setting $s=1$, we get $ \mathsf{GNDT}_1^{\mathrm{lb}}(\mu,\delta) = 1-\mu$.
Moreover, from \cite{Maddah-Ali2015a}, we know that $\mathsf{GNDT}^{\mathrm{ub}}_{\mathrm{D}}(\mu, \delta) \leq \mathsf{GNDT}^{\mathrm{ub}}_{\mathrm{D}}(\mu,1) \leq \frac{1}{\mu}-1$.
Hence, ${\mathsf{GNDT}^{\mathrm{ub}}_{\mathrm{D}}(\mu,\delta)}/{\mathsf{GNDT}_1^{\mathrm{lb}}(\mu,\delta)} \leq 12$.
This concludes the proof.
\section{Proof of Lemma \ref{lemma_ineq_decen}}
\label{sec:proof_lemma_ineq_decen}
It readily seen from the definition of $u$ in \eqref{eq:eq_intr_u}
that \mbox{$\mathsf{GNDT}_{\mathrm{D}}(\mu,1) =\mathsf{GNDT}^{\mathrm{ub}}_{\mathrm{D}}(\mu,1)$}.
It is also easy to verify that $\mathsf{GNDT}_{\mathrm{D}}(\mu,0) =\mathsf{GNDT}^{\mathrm{ub}}_{\mathrm{D}}(\mu,0)$
and $\mathsf{GNDT}^{\mathrm{ub}}_{\mathrm{D}}(1,\delta) = \mathsf{GNDT}_{\mathrm{D}}(1,\delta)=0$.
Therefore, we focus on $\delta \in (0,1)$ and $\mu \in [0,1)$.
We define $b_m$,  $m \in\{0,1,\ldots,K-1\}$, such that
\begin{equation}
b_m=\frac{K \binom{K-1}{m} \mu^ m (1-\mu)^{K-m}}{K ( 1-\mu)}.
\end{equation}
It can be shown that $\sum_{m=0}^{K-1} {b_m}=1$ as follows
\begin{align}
\sum_{m=0}^{K-1} {b_m} & =  \frac{1}{K ( 1-\mu)} \sum_{m=0}^{K-1} {K \binom{K-1}{m} \mu^m (1-\mu)^{K-m}}   \\
& \label{eq:proof_sum_bm} = \sum_{m=0}^{K-1} {\binom{K-1}{m} \mu^m (1-\mu)^{K-1-m}} = 1
\end{align}
where \eqref{eq:proof_sum_bm} follows from the binomial identity\footnote{Recall that the binomial identity is given by $(a+b)^n = \sum_{r=0}^n {\binom{n}{r}a^rb^{n-r}}$.}.
Hence, the inequality in  \eqref{eq:decent_inequality_1} is equivalently written as
\begin{align}
\label{eq:decent_int_ineq}
\sum_{m=0}^{K-1}\frac{b_m}{K(1-\delta)+(1+m)\delta} & \leq \frac{1}{K (1-\delta) + (1+u) \delta} \\
\label{eq:decent_int_ineq_2}
\Rightarrow \sum_{m=0}^{K-1}\frac{b_m}{c_m + v} & \leq \frac{1}{ \tilde{c} + v }.
\end{align}
where $v \triangleq K(1-\delta)$, $c_m\triangleq (1+m)\delta$ and $\tilde{c}\triangleq (1+u)\delta$.
By rearrangement of \eqref{eq:decent_int_ineq_2}, we obtain
\begin{equation} \label{eq:decent_inter_3}
(\tilde{c}'+1)\sum_{m=0}^{K-1}{\frac{b_m}{c_m'+1}} \leq 1.
\end{equation}
where  $\tilde{c}'=\tilde{c}/v$ and $c_m'=c_m/v$
By the definition of $u$ in \eqref{eq:eq_intr_u}, for any $\delta \in (0,1)$, we have
\begin{align} \label{eq:eq_1}
\sum_{m=0}^{K-1}\frac{b_m}{(1+m)\delta} & = \frac{1}{ (1+u)\delta} \\
\label{eq:eq_2}
\Rightarrow  \sum_{m=0}^{K-1}\frac{b_m}{c_m} & = \frac{1}{\tilde{c}} \\
\label{eq:eq_3}
\Rightarrow  \sum_{m=0}^{K-1}{\frac{b_m}{c_m'}} & =\frac{1}{\tilde{c}'}.
\end{align}
By plugging $\tilde{c}'$ from \eqref{eq:eq_3} into \eqref{eq:decent_inter_3}, we obtain
\begin{equation} \label{eq:decent_inter_4}
\left(\frac{1}{\sum_{m=0}^{K-1}{\frac{b_m}{c_m'}}}+1 \right)\sum_{m=0}^{K-1}{\frac{b_m}{c_m'+1}} \leq 1.
\end{equation}
Hence, showing that \eqref{eq:decent_inter_4} holds implies that \eqref{eq:decent_inequality_1} holds for
$\delta \in (0,1)$.
This is shown next.

Let us define the function $f(v) = \frac{v}{1+v}$, which is concave in $\mathbb{R}_{+} \setminus \{0\}$.
Moreover, consider the points $\big\{ \frac{1}{c_0'},\ldots,\frac{1}{c_{K-1}'} \big\} $ in $ \mathbb{R}_{+} \setminus \{0\}$.
From $\sum_{m=0}^{K-1}{b_m} = 1$, which is obtained from (\ref{eq:proof_sum_bm}), and by applying Jensen's inequality, we have
\begin{align}
& \sum_{m=0}^{K-1}{b_m f \left(\frac{1}{c_m'} \right)}  \leq f \left( \sum_{m=0}^{K-1} \frac{b_m}{c_m'}   \right) \\
\Rightarrow&
\sum_{m=0}^{K-1}{b_m \frac{1}{c_m' + 1}}  \leq \frac{\sum_{i=1}^n{\frac{b_m}{c_m'}}}{\sum_{m=0}^{K-1}{\frac{b_m}{c_m'}}+1} \\
\Rightarrow &
\left( \frac{\sum_{m=0}^{K-1}{\frac{b_m}{c_m'}}+1}{\sum_{m=0}^{K-1}{\frac{b_m}{c_m'}}} \right)  \left( \sum_{m=0}^{K-1}{\frac{b_m}{c_m' + 1}} \right)  \leq 1 \\
\Rightarrow &
\left(\frac{1}{\sum_{m=0}^{K-1}{\frac{b_m}{c_m'}}}+1 \right) \sum_{m=0}^{K-1}{\frac{b_m}{c_m'+1}} \leq 1
\end{align}
which is the inequality in (\ref{eq:decent_inter_4}). This concludes the proof.
\section{Proof of Corollary \ref{corollary_cen_decen}}
\label{appendix:proof_corollary}
First, for $\mu = 0$ we have that $\mathsf{GDoF}_{\mathrm{C}}(0,\delta) = \mathsf{GDoF}_{\mathrm{D}}(0,\delta) = K(1-\delta) + \delta$, while
for $\mu = 1$ we have that $\mathsf{GNDT}_{\mathrm{C}}(1,\delta) = \mathsf{GNDT}_{\mathrm{D}}(1,\delta) = 0$.
Therefore, we  focus on $\mu \in (0,1)$ in what follows.
The multiplicative factor of $1.5$ in \eqref{eq:cen_decen_corollary} can be shown by considering the three following cases:
\begin{enumerate}
\item $K\geq3$:
	From Theorem \ref{th_cen},  it follows that $\mathsf{GDoF}_{\mathrm{C}}(\mu,\delta)$
	is bounded above by
	\begin{align} \label{eq:GDoF_c_upperbound}
		\mathsf{GDoF}_{\mathrm{C}}(\mu,\delta) & \leq (1-\delta) \frac{K}{1-\mu}+\delta \frac{1+K\mu}{1-\mu}
	\end{align}
	where \eqref{eq:GDoF_c_upperbound} holds with equality
    for  $\mu \in \{0,\frac{1}{K},\frac{2}{K}, \dots, \frac{K-1}{K}\}$, as expressed in \eqref{eq:GDoF_cen_2}.
    For the remaining points in $\mu \in [0,1]$, the achievable GDoF upper bound in \eqref{eq:GDoF_c_upperbound} follows from
    \begin{equation}
	\mathsf{GNDT}_{\mathrm{C}}(\mu,\delta) \geq \frac{K(1-\mu)}{K(1-\delta) + (1+K\mu)\delta}
	\end{equation}
    which in turn holds as $\frac{K(1-\mu)}{K(1-\delta) + (1+K\mu)\delta}$ is convex in $\mu$
    and $\mathsf{GNDT}_{\mathrm{C}}(\mu,\delta)$ is the lower convex envelope (see \eqref{eq:T_cen}).
    From Lemma \ref{lemma_ineq_decen}, a lower bound for $\mathsf{GDoF}_{\mathrm{D}}(\mu,\delta)$ is given by
	\begin{align} \label{eq:upper_coroll_dec}
		\mathsf{GDoF}_{\mathrm{D}}(\mu,\delta) & \geq (1- \delta) \frac{K}{1-\mu}+ \delta \frac{1+u}{1-\mu}
	\end{align}
	where $1 + u =\frac{K \left(1-\mu \right)}{\mathsf{GNDT}_{\mathrm{D}}(\mu,1)}$ from \eqref{eq:eq_intr_u}.
	From \cite{Maddah-Ali2015a}, we know that $\mathsf{GNDT}_{\mathrm{D}}(\mu,1)$ can be written as
	\begin{equation}
	\mathsf{GNDT}_{\mathrm{D}}(\mu,1) = \frac{1-\mu}{\mu} \left( 1- (1-\mu)^K \right).
	\end{equation}
	It follows that $1+ u $ is given by
	\begin{equation} \label{eq:def_1_plus_mu}
	1+ u  = \frac{K\mu}{1- (1-\mu)^K}.
	\end{equation}
	From \eqref{eq:GDoF_c_upperbound} and \eqref{eq:upper_coroll_dec}, the ratio between $\mathsf{GDoF}_{\mathrm{C}}(\mu,\delta)$ and $\mathsf{GDoF}_{\mathrm{D}}(\mu,\delta)$ is bounded above as
	\begin{equation} \label{eq:cen_decen_ratio}
	\frac{\mathsf{GDoF}_{\mathrm{C}}(\mu,\delta)}{\mathsf{GDoF}_{\mathrm{D}}(\mu,\delta)} \leq \frac{(1-\delta) K+\delta (1+K\mu)}{(1- \delta) K+ \delta (1+u)} \leq \frac{1+K\mu}{1+u}.
	\end{equation}
where the rightmost inequality in \eqref{eq:cen_decen_ratio} follows from $u \leq K \mu$, which in turn is obtained from
Lemma \ref{lemma_bound_u} in Appendix \ref{sec:proof_order_opt_T}.
	By plugging \eqref{eq:def_1_plus_mu} into \eqref{eq:cen_decen_ratio}, we obtain
	\begin{equation} \label{eq:final_bounding_K_geq_3}
	\frac{1+K\mu}{1+u} = \frac{1+K\mu}{K\mu} \left(  1 - (1-\mu)^K \right) \leq 1.5
	\end{equation}
   where the bound by $1.5$ follows directly from  \cite[Lem. 1]{Yan2016}.
	\item $K=2$: For this case, we consider the two following subcases:
	\begin{itemize}
		\item $\mu \in (0,1/2]$: For this interval, we employ the same bounding techniques used for the case $K \geq 3$.
		Hence, from (\ref{eq:cen_decen_ratio}) and (\ref{eq:final_bounding_K_geq_3}) we obtain
		\begin{equation} \label{eq:case_K_2}
		\frac{\mathsf{GDoF}_{\mathrm{C}}(\mu,\delta)}{\mathsf{GDoF}_{\mathrm{D}}(\mu,\delta)} \leq \frac{1+2\mu}{2\mu}
\left(  1 - (1-\mu)^2 \right).
		\end{equation}
		It is readily seen that the right-hand-side of (\ref{eq:case_K_2}), which we denote as $g(\mu)$, is a concave parabola
		with a maximum at $\mu = 3/4$.
		Given the symmetry of the parabola, it follows that that $g(\mu) \leq g(1/2) = 1.5$ for  $\mu \in (0,1/2]$.
		\item $\mu \in [1/2,1)$: For this interval, the bounding techniques used for the case $K \geq 3$ are loose.
		Alternatively, it can be easily shown from  Theorem \ref{th_cen} that
        $\mathsf{GDoF}_{\mathrm{C}}(\mu,\delta) = \frac{2}{1-\mu}$.
		Combining this with the upper bound for $\mathsf{GDoF}_{\mathrm{D}}(\mu,\delta)$ in (\ref{eq:upper_coroll_dec}), we obtain
		\begin{equation} \label{eq:case_K_2_2}
		\frac{\mathsf{GDoF}_{\mathrm{C}}(\mu,\delta)}{\mathsf{GDoF}_{\mathrm{D}}(\mu,\delta)} \leq \frac{2}{2(1- \delta) +  (1+u)\delta}
\leq \frac{2}{1+u}
		\end{equation}
		where the rightmost inequality in \eqref{eq:case_K_2_2}  follows from the fact that $1 + u \leq 2$, which can be easily shown.
        By plugging \eqref{eq:def_1_plus_mu} into \eqref{eq:case_K_2_2}, we obtain
		\begin{equation}
		\frac{2}{1+u} = \frac{1}{\mu} \left(  1 - (1-\mu)^2 \right) = 2 - \mu.
		\end{equation}
		It is readily seen that $2-\mu \leq 1.5$ for $\mu \in [1/2,1)$.
	\end{itemize}
\item Case $K=1$: In this case we have $\mathsf{GNDT}_{\mathrm{C}}(\mu,\delta) = \mathsf{GNDT}_{\mathrm{D}}(\mu,\delta) = 1 -\mu$, hence (\ref{eq:cen_decen_corollary}) holds.
\end{enumerate}
From the above three cases, the proof is complete.
It is worthwhile highlighting that for the case $K=2$, $\delta =1$ and $\mu = 1/2$,
we have ${\mathsf{GDoF}_{\mathrm{C}}(\mu,\delta)}/{\mathsf{GDoF}_{\mathrm{D}}(\mu,\delta)} = 1.5$.
Therefore, $1.5$ is in fact the tightest possible upper bound for  ${\mathsf{GDoF}_{\mathrm{C}}(\mu,\delta)}/{\mathsf{GDoF}_{\mathrm{D}}(\mu,\delta)}$.
\section*{Acknowledgment}
The authors would like to thank the anonymous reviewers for their valuable comments.
The authors are also grateful to Reviewer 1 for suggesting a shorter and more direct proof for Lemma \ref{lemma:AIS}.
\color{black}
\bibliographystyle{IEEEtran}
\bibliography{References}

\end{document}